\documentclass[a4paper,UKenglish,cleveref, autoref]{lipics-v2019}

\title{On the impact of treewidth in the computational complexity of freezing dynamics} 

\titlerunning{On the impact of treewidth in the computational complexity of freezing dynamics}
\author{Eric Goles}{Facultad de Ingeniería y Ciencias, Universidad Adolfo Ibáñez, Santiago, Chile}{}{}{}

\author{Pedro Montealegre}{Facultad de Ingeniería y Ciencias, Universidad Adolfo Ibáñez, Santiago, Chile}{}{}{}

\author{Martín Ríos-Wilson}{Departamento de Ingeniería Matemática, FCFM, Universidad de Chile, Santiago, Chile. \and Aix Marseille Univ, Université de Toulon, CNRS, LIS, Marseille, France}{}{}{}

\author{Guillaume Theyssier}{Aix-Marseille Université, CNRS, I2M (UMR 7373), Marseille, France}{}{}{}

\authorrunning{E. Goles, P. Montealegre, M. Ríos-Wilson, G. Theyssier}

\Copyright{Eric Goles, Pedro Montealegre, Martín Ríos-Wilson, Guillaume Theyssier}



\nolinenumbers 

\hideLIPIcs  

\newtheorem{problem}[theorem]{Problem}
\usepackage{tikz}

\newcommand\N{\mathbb{N}}
\DeclareMathOperator{\tw}{tw}
\newcommand{\cO}{\mathcal{O}}
\newcommand{\Sol}{\textrm{Sol}}
\newcommand{\Times}{\textsc{Times}}
\newcommand{\States}{\textsc{States}}

\newcommand{\Level}{\textsc{Level}}

\begin{document}

\maketitle

\begin{abstract}
  An automata network is a network of entities, each holding a state from a finite set and evolving according to a local update rule which depends only on its neighbors in the network's graph. It is freezing if there is an order on states such that the state evolution of any node is non-decreasing in any orbit. They are commonly used to model epidemic propagation, diffusion phenomena like bootstrap percolation or cristal growth. 
  
  In this paper we establish how alphabet size, treewidth and maximum degree of the underlying graph are key parameters which influence the overall computational complexity of finite freezing automata networks. First, we define a general decision problem, called \textsf{Specification Checking Problem}, that captures many classical decision problems such as prediction, nilpotency, predecessor, asynchronous reachability.
   
   Then, we present a fast-parallel algorithm that solves the general model checking problem when the three parameters are bounded, hence showing that the problem is in NC. Moreover, we show that the problem is in XP on the parameters tree-width and maximum degree.

   Finally, we show that these problems are hard from two different perspectives. First, the general problem is W[2]-hard when taking either treewidth or alphabet as single parameter and fixing the others. Second, the classical problems are hard in their respective classes when restricted to families of graph with sufficiently large treewidth. Moreover, for prediction, predecessor and asynchronous reachability, we establish the hardness result with a fixed set-defiend update rule that is universally hard on any input graph of such families.

\end{abstract}
\section{Introduction}

An automata network is a network of $n$ entities, each holding a state from a finite set $Q$ and evolving according to a local update rule which depends only on its neighbors in the network's graph. More concisely, it can be seen as a dynamical system (deterministic or not) acting on the set $Q^n$. The model can be seen as a non-uniform generalization of (finite) cellular automata. Automata networks have been used as modelization tools in many areas \cite{Goles90} and they can also be considered as a distributed computational model with various specialized definitions like in \cite{WR79i,WR79ii}. 

An automata network is \emph{freezing} if there is an order on states such that the state evolution of any node is non-decreasing in any orbit. Several models that received a lot of attention in the literature are actually freezing automata networks, for instance: bootstrap percolation which has been studied on various graphs \cite{Amini_2014,Balogh_2012,Balogh_2005,Holroyd03}, epidemic \cite{Fuentes} or forest fire \cite{forestfire} propagation models \footnote{They are discrete counterparts of the family of spatial SIR/SEIR models \cite{KMcK27} (or other variants) which are sadly famous amid the actual COVID-19 pandemic.}, cristal growth models \cite{ulam,Gravner98} and more recently self-assembly tilings \cite{Winslow16}. On the other hand, their complexity as computational models has been studied from various point of view: as language recognizers where they correspond to bounded change or bounded communication models \cite{vollmar81,KutribM10a,Carton2018}, for their computational universality \cite{arxivOllinger19,GolOlThey15,BeckerMOT18}, as well as for various associated decision problems \cite{GolesMMO17,GolesMMO18,Goles_2013,goles2019complexity}.

A major topic of interest in automata networks theory is to determine how the network graph affects dynamical or computational properties \cite{Gad18b,Goles_2013}. In the freezing case, it was for instance established that one-dimensional freezing cellular automata, while being Turing universal via Minsky machine simulation, have striking computational limitations when compared to bi-dimensional ones: they are NL-predictable (instead of P-complete) \cite{arxivOllinger19,GriMoo96,vollmar81}, can only produce computable limit fixed points starting from computable initial configurations (instead of non-computable ones starting from finite configurations) \cite{arxivOllinger19}, and have a polynomial time decidable nilpotency problem (instead of uncomputable) \cite{arxivOllinger19}.

The present paper aims at understanding what are the key parameters which influence the overall computational complexity of finite freezing automata networks. A natural first parameter is the alphabet size, as  automata networks are usually considered as simple machines having a number of states that is independent of the size of the network.  For the same reasons, a second parameter that we consider is the maximum degree of the network, as a simple machine might not be able to handle the information incoming from a large number of neighbors. Finally the results mentioned earlier show a gap between bi-dimensional grids and one-dimensional grids (i.e. paths or rings). Since Courcelle's theorem on MSO properties \cite{Courcelle_1990}, graph parameters like treewidth \cite{Robertson_1986} are used to measure a sort of distance to a grid. Indeed, it is known that paths or rings have constant treewidth, and the treewidth of a graph is polynomially related to the size of its largest grid minor \cite{Chekuri_2016}. Therefore treewidth is a natural parameter for our study. 



\emph{A canonical model checking problem to capture many classical dynamical problems.} We define a general model checking problem \textbf{SPEC}, that asks whether a given freezing automata networks has an orbit that satisfies a given set of local constraints on the trace at each node (Problem~\ref{prob:spectchecking}). It takes advantage of the sparse orbits of freezing automata networks (a bounded number of changes per node in any orbit) which allow to express properties in the temporal dimension in an efficient way. We show thanks to a kind pumping lemma on orbits (Lemma~\ref{lem:trace-length}) that it captures many standard problems in automata network theory, among which we consider four ones: prediction \cite{GolesMMO17,GolOlThey15,GriMoo96}, nilpotency \cite{Richard_2019,Gadouleau_2016,kari92}, predecessor \cite{Kawachi2019,greenlaw1995limits} and asynchronous reachability \cite{Dennunzio2017}. Note that since Boolean circuits are easily embedded into freezing automata networks, our framework also includes classical problems on circuit: circuit value problem is a sub-problem of our prediction problem (see Theorem~\ref{theo:pcompleteprediction}) and SAT is a sub-problem of our nilpotency problem (see Remark~\ref{rem:nilpotencyontrees} and Theorem~\ref{thm:nilpotency}).

\emph{Fast parallel algorithm.} We then present a NC algorithm that solves our general model checking problem \textbf{SPEC} on any freezing automata network with bounded number of states and graph with bounded degree and bounded treewidth (Theorem~\ref{lem:algospeci}). It solves in particular the four canonical problems above in NC for such graphs, as well as circuit value problem and SAT (see \cite{Szeider_2004,Buss_1987} for better known results for these specific problems). Note that our algorithm is uniform in the sense that, besides the graph, both the automata network rule and the constraint to test are part of the input and not hidden in an expensive pre-processing step. As suggested above, temporal traces of the evolution of a bounded set of nodes have a space efficient representation. However, it is generally hard to distinguish real orbits projected on a set of nodes from locally valid sequences of states that respect the transition rule for these nodes. Our algorithm exploits bounded treewidth and bounded degree to solve this problem via dynamic programming for any finite set of nodes. In the deterministic case, our algorithm can completely reconstruct the orbit from the initial configuration. 

\emph{Hardness results.} In light of Courcelle's theorem, one might think that our algorithm solving \textbf{SPEC} could be directly obtained (or even improved) by simply expressing the problem in MSOL \cite{Courcelle_1990}. We show that previous statement is impossible, unless $W[2] = FPT$. More precisely, we show that the version of our model checking problem where treewidth is the unique parameter and alphabet and degree are fixed is $W[2]$-hard (Corollary~\ref{cor:w2fixedalphabet}), and thus is not believed to be fixed parameter tractable. We obtain a similar result on a slight variation of our problem when alphabet is the unique parameter and treewidth and degree are fixed (Corollary~\ref{cor:w2fixedtreewidth}). 

Finally, we prove that the four problems mentioned before, (namely prediction, nilpotency, predecessor and asynchronous reachability) are complete in there respective class (respectively P-Complete and coNP- Complete for the first two, and NP-Complete for the later) when we restrict the input graphs to a constructible family of sufficiently large treewidth (Theorems~\ref{thm:nilpotency}, \ref{teo:predeasyncNP} and \ref{theo:pcompleteprediction}). To do so, we rely on an efficient algorithm to embed arbitrary (but polynomially smaller) digraph into our input graph (Lemma~\ref{lem:routing}), which relies on polynomial perfect brambles that can be efficiently found in graphs with polynomially large treewidth \cite{KreutzerT10} (here by polynomial we mean ${\Omega(n^\alpha)}$ for some positive real number $\alpha$). This embedding allows to simulate a precise dynamics on the desired digraph inside the input graph and essentially lifts us from the graph family constraint as soon as the treewidth is large enough. Moreover, for problems prediction (Theorem~\ref{theo:pcompleteprediction}), predecessor and asynchronous reachability (Theorem~\ref{teo:predeasyncNP}) we achieve the hardness result with a fixed uniform set-defined rule (\textit{i.e.} a rule that change the state of each node depending only on the set of states seen in the neighborhood) which is not part of the input. This shows that there is a uniform isotropic universally hard rule for these problems, which makes sense for applications like bootstrap percolation, epidemic propagation or cristal growth where models are generally isotropic and spatially uniform.  

\section{Preliminaries}
Given a graph $G=(V,E)$ and a vertex $v$ we will call $N(v)$ to the neighborhood of $v$ and $\delta_v$ to the degree of $v$.  In addition, we define the closed neighborhood of $v$ as the set $N[v] = N(v) \cup \{v\}$ and we use the following notation $\Delta(G) = \max \limits_{v \in V} \delta_v$  for the maximum degree of $G$. We will use the letter $n$ to denote the order of $G$, i.e. $n = |V|$.  Also, if $G$ is a graph and the set of vertices and edges is not specified we use the notation $V(G)$ and $E(G)$ for the set of vertices and the set of edges of $G$ respectively. In addition, we will assume that if $G =(V,E)$ is a graph then, there exist an ordering of the vertices in $V$ from $1$ to $n$.  During the rest of the text, every graph $G$ will be assumed to be connected and undirected.  We define a \textit{class} or a \textit{family} of graphs as a set $\mathcal{G}=\{G_n\}_{n \geq 1}$ such that $G_n = (V_n,E_n)$ is a graph and $|V_n| = n$. 

\emph{Non-deterministic freezing automata networks.}
Let $Q$ be a finite set that we will call an \textit{alphabet}. We define a non-deterministic automata network in the alphabet $Q$ as a tuple $(G=(V,E),\mathcal{F}=\{F_v : Q^{N(v)} \to \mathcal{P}(Q) | v \in V\}))$ where $\mathcal{P}(Q)$ is the power set of $Q$. To every non-deterministic automata network we can associate a non-deterministic dynamics given by the global function $F: Q^n \to \mathcal{P}(Q^n) $ defined by $F(x) = \{ x \in Q^n | x_v \in F_v(x),\forall v\}.$
\begin{definition}
	Given a a non-deterministic automata network $(G,\mathcal{F})$ we define an orbit of a configuration $x \in Q^n$ at time $t$ as a sequence $(x_s)_{0\leq s\leq t}$ such that $x_0 = x$ and $x_{s} \in F(x_{s-1}).$ In addition, we call the set of all possible orbits  at time $t$ for a configuration $x$ as $\mathcal{O} (x,t)$. Finally, we also define the set of all possible orbits at time $t$ as $\mathcal{O}(\mathcal{A},t
	) = \bigcup \limits_{x \in Q^n} \mathcal{O}(x,t)$
\end{definition}
We say that a non-deterministic automata network  $(G,\mathcal{F})$ defined in the alphabet $Q$ satisfies the \textit{freezing property} or simply that it is \textit{freezing} if there exists a partial order $\leq$ in $Q$ such that for every $t  \in \N$ and for every orbit $y = (x_s)_{0\leq s \leq t} \in \mathcal{O}(\mathcal{A},t)$ we have that $x^i_s \leq x^i_{s+1}$  for every $0 \leq s \leq t$ and for every $0\leq i \leq n.$
Let $y = (x_s)_{0\leq s \leq t}$  be an orbit for  a non-deterministic automata network $(G,\mathcal{F})$ and $S \subseteq V$ we define the restriction of $y$ to $S$ as the sequence $z \in (Q^t)^{|S|}$ such that $x^v_s = z^v_s$ for every $ v \in V$ and we note it $y|_S.$ In the case in which $S = \{v\}$ we simply write $y_v$ in order to denote the restriction of $y$ to the singleton $\{v\}$
\begin{definition}
	Given a a non-deterministic automata network $(G,\mathcal{F})$ and a set $S \subseteq V$,  we define the set of $S$-restricted orbits as the set $\mathcal{T}(S,t) = \{z = (x_s)_{s \leq t} \in Q^{|S|} \text{ } | \text{  } \exists y \in \mathcal{O}(t): y|_{S} = z\}$. When ${S=\{v\}}$ we simply write ${\mathcal{T}(v,t)}$ for ${\mathcal{T}(\{v\},t)}$.
\end{definition}
During the rest of the text and we use the notation $z = x|_{S}$. Finally, if $\mathcal{A} = (G,\mathcal{F})$ is a non-deterministic freezing automata network such that for every $v \in V(G),$  $F_v \in \mathcal{F}$ is such that $|F_v(x)| = 1,$ for all $x \in  Q^{N(v)}$ then, we say that $\mathcal{A}$ is deterministic and view local rules as maps ${F_v:Q^{N(v)}\rightarrow Q}$ to simplify notations.

\emph{Tree decompositions and treewidth.}
Let $G = (V,E)$ be a connected graph. A subgraph $P$ of $G$ is said to be a path if $V(P) = \{v_1,\hdots,v_k\}$ where every $v_i$ is different and $E(P) = \{v_1v_2 , v_2v_3\hdots, v_{k-1}v_k \}$.  We define the length of a path $P$ in $G$ as the number of edges of $P$. Given two vertices $u,v \in V(G)$ we say that $P$ is a v-u path if $v_1 = v$ and $v_k = u$ We say that $P$ is a cycle if $k \geq 3$ and $v_k = v_1$. We say that $G$ is a tree-graph or simply a tree if it does not have cycles as subgraphs. Usually, we will distinguish certain node in $r \in V(G)$ that we will call the root of $G$. Whenever $G$ is a tree and there is a fixed vertex $r \in V(G)$ we will call $G$ a rooted tree-graph. In addition, we will say that $v \in V(G)$ is a leaf if $\delta_v= 1$. Straightforwardly the choice of $r$ induces a partial order in the vertices of $G$ given by the distance (length of the unique path) between a node $v \in V(G)$ and the root $r$. We define the height of $G$ (and we write it as  $h(G)$) as the longest path between a leaf and $r$.  We say that a node $v$ is in the $(h(G)-k)$-th level of a tree-graph $G$ if the distance between $v$ and $r$ is $k$ and we write $v \in \mathcal{L}_{h(G)-k}$. We will call the \textit{children} of a node $v \in \mathcal{L}_k$ to all $w \in N(v)$ such that $w$ is in level $k-1$.

\begin{definition}
	Given a graph $G= (V,E)$ a tree decomposition is pair $ \mathcal{D} = (T,\Lambda)$ such that $T$ is a tree graph and $\Lambda$ is a family of subsets of nodes  $\Lambda = \{X_t \subseteq V | \text{ } t \in V(T) \}$, called bags, such that:
	\begin{itemize}
		\item Every node in $G$ is in some $X_t$, i.e: $\bigcup \limits_{t \in V(T)} X_t = V$
		\item For every $e=uv \in E$ there exists $t \in V(T)$ such that $u,v \in X_t$
		\item For every $u,v \in V(T)$ if $w \in V(T)$ is in the $v$-$y$ path in $T$, then $X_u \cap X_v \subseteq X_w$
	\end{itemize}
\end{definition}
We define the width of a tree decompostion $\mathcal{D}$ as the amount $ \text{width} (\mathcal{D}) = \max  \limits_{t \in V(T)}|X_t|-1$.  Given a graph $G= (V,E)$, we define its treewidth as the parameter  $\text{tr}(G) = \min \limits_{\mathcal{D}} \text{width}(\mathcal{D})$. In other words, the treewidth is the minimum width of a tree decomposition of $G$. Note that, if $G$ is a connected graph such that $|E(G)| \geq 2$ then, $G$ is a tree  if and only if $\text{tw}(G) = 1$. 

It is well known that, given an arbitrary graph $G$, and $k \in \mathbb{N}$, the problem of deciding if $\text{tw}(G) \leq k$ is $\textbf{NP}$-complete \cite{Arnborg1987}. Nevertheless, if $k$ is fixed, that is to say, it is not part of the input of the problem then, there exist efficient algorithms that allow us to compute a tree-decomposition of $G$. More precisely, it is shown that for every constant $k \in \N$ and a graph $G$ such that $\text{tw}(G) \leq k$, there exist a log-space algorithm that computes a tree-decomposition of $G$ \cite{Elberfeld2010}. In addition,  in Lemma 2.2 of \cite{Bodlaender1998} it is shown that given any tree decomposition of a graph $G$,  there exist a fast parallel algorithm that computes  a slightly bigger width binary tree decomposition of $G$. More precisely, given a tree decomposition of width $k$, the latter algorithm computes a binary tree decomposition of width at most $3k +2.$  We outline these results in the following proposition:
\begin{proposition}
	\label{prop:treedecompose}
	Let $n\geq 2, k \geq 1$ and let $G=(V,E)$ with $|V| = n$ be a graph such that $\text{tw}(G)\leq k$. There exists a  CREW PRAM algorithm  using $\mathcal{O}(\log^2 n)$ time, $n^{\mathcal{O}(1)}$ processors  and $\mathcal{O}(n)$ space  that computes a binary treewidth decomposition of width  at most $3k+2$ for $G$.
\end{proposition}
We now present basic concepts in parameterized complexity that we will be using during this paper (see~\cite{Downey_2013} for more details and context). A parameterized language is defined by $L \subset \{0,1\}^* \times \mathbb{N}.$ Whenever we take an instance $(x,k)$ of a parameterized problem we will call $k$ a \textit{parameter}. The objective behind parameterized complexity is to identify which are the key parameters in an intractable problem  that make it hard.

We say that a parameterized language $L$ is \emph{slice-wise polynomial} if $(x,k) \in L$ is decidable in polynomial time for every fixed $k\in \mathbb{N}$. More precisely, when $(x,k)\in L$ can be decided in time $|x|^{f(k)}$ for some arbitrary function $f$ depending only on $k$. The class of slice-wise polynomial parameterized languages is called XP. 

An important subclass of XP is the set of parameterized languages $L$ that are \textit{fixed-parameter tractable}, denoted FPT.  A parameterized language $L$ is in FPT if there exist an algorithm deciding if $(x,k) \in L$ in time $f(k)|x|^{\cO(1)}$ where $f$ an arbitrary function depending only in $k$. It is known that XP is not equal to FPT, however showing that some problem in XP is not in FPT seems currently out of reach for many natural examples.
As in other domains of complexity a hardness theory has been developed relying on the following notion of reduction. Given two parameterized languages $L_1, L_2$ we say that $L_1$ is FPT reducible to $L_2$ (and we write this as $L_1 \leq_{\textrm{FPT}} L_2)$  if there exist some functions $r,s: \N \to \N$ and $M: \{0,1\}^* \times \mathbb{N} \to \{0,1\}^*$ such that for each instance $(x,k)$ of $L_1$, $M$ is computable in time $s(k)|x|^c$ for some constant $c$ and $(x,k) \in L_1$ if and only if $(M(x), r(k)) \in L_2$.  A hierarchy of parameterized languages has been defined, called $W$-hierarchy, that contain a countable sequence of classes of parameterized languages, namely $\textrm{W}[1]$, $\textrm{W}[2], \textrm{W}[3] \dots $ such that $\textrm{FPT} \subset \textrm{W}[1] \subset, W[2] \subset \hdots \textrm{XP}$ and it is conjectured that are inclusions are proper. We don't give the formal definition these classes and we refer to \cite{Downey_2013} for more details. For our purposes, it is enough to know that $k$-\textsf{Dominating-Set} (i.e. finding a set of $k$ nodes that intersects the neighborhood of any node) is $W[2]$-hard and that a parameterized language is $W[2]$-hard if there is an FPT-reduction from $k$-\textsf{Dominating-Set} to it.

\paragraph*{Fast-parallel sub-routines}

Finally, we cite the following results that we use as some kind of toolbox for the proof of our main results:

\begin{proposition}[Prefix-sum algorithm, \cite{jaja}]
	\label{prop:prefixsum}
	Then the following problem can be solved by a CREW PRAM machine with $p = \mathcal{O}(n)$ processors in time $\mathcal{O}(\log n)$: 	Given $A = \{x_1,\hdots, x_n\}$ be a finite set, $k\leq n $ and $\oplus$ be a binary associative operation in $A,$  compute $\oplus_{i = 1}^{k} x_i $
\end{proposition}
\begin{proposition}[{\cite[Theorem 5.3]{Reingold2008}}]
	Let $n \in \mathbb{N}$. The following problem can be solved in space $\mathcal{O}(\log n)$: given an undirected graph $G=(V,E)$ with $|V| = n$, $s, t \in V$ find a path from $s$ to $t$ and if there exists such a path, return the path as an output.
	\label{prop:pathsLOG}
\end{proposition}
\begin{proposition}[{\cite[Theorem 3]{goldberg1987parallel}}]
	Let $\Delta \in \mathbb{N}$.  The following problem can be solved in  time $\mathcal{O}(\Delta \log ( \Delta +\log^*n) )$ by an EREW PRAM: given a graph $G = (V,E)$  such that $\Delta(G) \leq \Delta$ finding a $\Delta +1$ coloring of $G$.
	\label{prop:colNC}
\end{proposition}

\section{Localized Trace Properties}

In this section we formalize the general decision problem we consider on our dynamical systems. Freezing automata network have temporally sparse orbits, however the set of possible configurations is still exponential. Our formalism takes this into account by considering properties that are spatially localized but without restriction in their temporal expressive power. More precisely, we introduce the concept of a specification, in an attempt of generalizing the notion of parallelizable partial information regarding possible orbits of a freezing automata network.
\begin{definition}
	Let $t$ be a natural number and $\mathcal{A}=(G=(V,E),\mathcal{F})$ a non-deterministic freezing automata network in some partially-ordered alphabet $Q$. 
	A $(Q,t, \mathcal{A})$-specification (or simply a $t$-specification when the context is clear) is a function $\mathcal{E}_t: V \to \mathcal{P}(Q^t)$ such that, for every $v\in V$, the sequences in $\mathcal{E}_t(v)$ are non-decreasing. 
\end{definition}

\emph{Pumping lemma on orbits.} The following lemma shows that for all freezing automata networks the set of orbits of any length restricted to a set of nodes is determined by the set of orbits of fixed (polynomial) length restricted to these nodes. Moreover, if the set of considered nodes is finite, then the fixed length can be chosen linear.

\begin{lemma}
  \label{lem:trace-length}
  Let $Q$ be an alphabet, $V$ a set of nodes with ${|V|=n}$ and ${U\subseteq V}$. Let $L=|U||Q|(|Q|n+1)$. Then if two non-deterministic freezing automata have the same set of orbits restricted to $U$ of length $L$ then they have the same set of orbits restricted to $U$ of any length.
\end{lemma}

\begin{proof}
	Any orbit restricted to $U$ of any length can be seen as a sequence of elements of $Q^U$ and, since the considered automata network is freezing, there are at most $|U||Q|$ changes in this sequence so that it can be written ${p_1^{t_1}p_2^{t_2}\cdots p_m^{t_m}}$ with ${m\leq |U||Q|}$, $p_i\in Q^U$ and ${t_i\in\N}$. The key observation is that ${p_1^{t_1}p_2^{t_2}\cdots p_{i-1}^{t_{i-1}}p_i^{|Q|n+1}p_{i+1}^{t_{i+1}}\cdots p_m^{t_m}}$ is a valid restricted orbit if and only if ${p_1^{t_1}p_2^{t_2}\cdots p_{i-1}^{t_{i-1}}p_i^{T}p_{i+1}^{t_{i+1}}\cdots p_m^{t_m}}$ is a valid restricted orbit for all ${T\geq|Q|n+1}$: this is because any sequence of ${|Q|n+1}$ configurations in any orbit must contain two consecutive identical configurations since ${|Q|n}$ is the maximal total number of possible state changes. From this we deduce that it is sufficient to know all the restricted orbits of the form ${p_1^{t_1}p_2^{t_2}\cdots p_m^{t_m}}$ with ${t_i\leq |Q|n+1}$ and ${m\leq |U||Q|}$ to know all restricted orbits of any length. The lemma follows.
\end{proof}

Note that as a consequence of the latter lemma, for any freezing non-deterministic automata network it suffices to consider $t$-specifications with $t$ being linear in the size of the interaction graph defining the network.

\emph{Specification checking problem.}
We observe also that the number of possible $t$-specifications can be represented in polynomial space (as a Boolean vector indicating the allowed $t$-specifications). Also, in the absence of explicit mention, all the considered graphs will have bounded degree $\Delta$ by default, so a freezing automata network rule can be represented as the list of local update rules for each node which are maps of the form ${Q^\Delta\to \mathcal{P}(Q)}$ whose representation as transition table is of size ${O\bigl(|Q|^{\Delta+1}\bigr)}$  .
The specification checking problem we consider asks whether a given freezing automata network verifies a given localized trace property on the set of orbits whose restriction on each node adheres to a given $t$-specification. In order to do that, we introduce the concept of a satisfiable $t$-specification
\begin{definition}
	Let $\mathcal{A}=(G,\mathcal{F})$ be a non-deterministic automata network and let $\mathcal{E}_t$ a $t$-specification. We say that $\mathcal{E}_t$ is satisfiable by $\mathcal{A}$ if there exists an orbit  $O \in \mathcal{O}(\mathcal{A},t)$ such that $O_v \in \mathcal{E}_t(v)$ for every $v \in V.$ 
\end{definition}
If $\mathcal{E}_t$ is a satisfiable $t$-specification for some automata network $\mathcal{A}$ we write $\mathcal{A} \models \mathcal{E}_t.$ We present now the $\textit{Specification checking problem}$ as the problem of verifying whether a given $t$-specification is satisfiable by some automata network $\mathcal{A}.$
\begin{problem}[Specification checking problem (\textsc{SPEC})]\ 
  \label{prob:spectchecking}
  \begin{description}
  \item[Parameters: ] alphabet $Q$, family of graphs $\mathcal{G}$ of max degree $\Delta$.
  \item[Input: ]\ 
    \begin{enumerate}
    \item a non-deterministic freezing automata network
      $\mathcal{A}=(G,F)$ on alphabet $Q$, with set of nodes
      $V$ and ${G\in\mathcal{G}}$;
      \item a time $t\in \mathbb{N}$.
    \item a $t$-specification $\mathcal{E}_t$
  \end{enumerate}
\item[Question: ] $\mathcal{A} \models \mathcal{E}_t$
  \end{description}
\end{problem}
We remark that it could be also possible to present some sort of universal version of the latter problem, in which we could ask not only if given $t$-specification $\mathcal{E}_t$ is satisfiable in the sense of checking for the \textbf{existence} of some orbit of the system verifying some property coded in $\mathcal{E}_t$ but checking if \textbf{every} orbit verifies the latter property.
%
%

\emph{Four canonical problems.} When studying a dynamical system, one is often interested in determining properties of the future state of the system given its initial state. In the context of automata networks, various decision problems have been studied where a question about the evolution of the dynamics at a given node is asked. Usually, the computational complexity of such problems is compared to the complexity of simulating the automata network. Roughly, one can observe that some systems are complex in some way if the complexity of latter problems are "as much as hard" as simply simulating the system.

\begin{problem}[Prediction problem]\ 
  \label{prob:prediction}
  \begin{description}
  \item[Parameters: ] alphabet $Q$, family of graphs $\mathcal{G}$ of max degree $\Delta$
  \item[Input: ]\ 
    \begin{enumerate}
    \item a deterministic freezing automata network
      $\mathcal{A}=(G,F)$ on alphabet $Q$, with set of nodes
      $V$ with ${n=|V|}$ and $G\in\mathcal{G}$;
    \item an initial configuration $c\in Q^V$;
    \item a node $v\in V$ and a time ${t\in\N}$;
    \item A $t$-specification $\mathcal{E}_t$ satisfying: for all $y \in \mathcal{E}_t(v), y_0 = c_v.$
  \end{enumerate}
\item[Question: ] $\mathcal{A} \models \mathcal{E}_t$
  \end{description}
\end{problem}

Note that this prediction problem is clearly a subproblem of \textsc{SPEC}. Also, observe that a specification allows us to ask various questions considered in the literature: what will be the state of the node at a given time \cite{GolOlThey15,GriMoo96}, will the node change its state during the evolution \cite{GolesMMO17,GolesMMO18,goles2019complexity}, or, thanks to Lemma~\ref{lem:trace-length}, what will be state of the node once a fixed point is reached \cite[section 5]{arxivOllinger19}.
Note that the classical circuit value problem for Boolean circuits easily reduces to the prediction problem above when we take $G$ to be the DAG of the Boolean circuit and choose local rules at each node that implement circuit gates. Theorem~\ref{theo:pcompleteprediction} below gives a much stronger result using such a reduction where the graph and the rule are independent of the circuit.

We now turn to the classical problem of finding predecessors back in time to a given configuration \cite{Kawachi2019,green97}.

\begin{problem}[Predecessor Problem]\ 
  \label{prob:predecessor}
  \begin{description}
  \item[Parameters: ] alphabet $Q$, family of graphs $\mathcal{G}$ of max degree $\Delta$
  \item[Input: ]\ 
    \begin{enumerate}
    \item a deterministic freezing automata network
      $\mathcal{A}=(G,F)$ on alphabet $Q$, with set of nodes
      $V$ with ${n=|V|}$ and $G\in\mathcal{G}$ ;
    \item a configuration $c\in Q^V$
    \item a time $t\in\N$
  \end{enumerate}
\item[Question: ] ${\exists y\in Q^V : F^t(y)=c}$?
  \end{description}
\end{problem}
Note that, analogously to the previous case,  the final configuration in the input can be given through a particular $t$-specification $\mathcal{E}_t$, such that for all $y \in \mathcal{E}_t(v): y_t = c$ for any $v \in V$. Thus, by considering  $\mathcal{E}_t$ we can see predecessor problem as a subproblem of \textsc{SPEC}.

Deterministic automata networks have ultimately periodic orbits. When they are freezing, any configuration reaches a fixed point. Nilpotency asks whether their is a unique fixed point whose basin of attraction is the set of all configurations. It is a fundamental problem in finite automata networks theory \cite{Richard_2019,Gadouleau_2016} as well as in cellular automata theory where the problem is undecidable for any space dimension \cite{kari92}, but whose decidability depends on the space dimension in the freezing case \cite{arxivOllinger19}.

\begin{problem}[Nilpotency problem]\ 
  \label{prob:nilpotency}
  \begin{description}
  \item[Parameters: ] alphabet $Q$, family of graphs $\mathcal{G}$ of max degree $\Delta$
  \item[Input: ] a deterministic freezing automata network
      $\mathcal{A}=(G,F)$ on alphabet $Q$, with set of nodes
      $V$ and $G\in\mathcal{G}$;
\item[Question: ] is there $t\geq 1$ such that ${F^t(Q^V)}$ is a singleton?
  \end{description}
\end{problem}

In this case, it is not clear that Nilpotency is actually a subproblem of $\textsc{SPEC}$. However, we will show that we can solve Nilpotency by solving a polynomial amount of instances (linear on the size of the interaction graph of the network $|G|$) for \textsc{SPEC} in parallel. More precisely, we show that there exist a NC Turing reduction. In order to do that note first that we can use Lemma~\ref{lem:trace-length} to fix ${t=\lambda(n)}$, where ${\lambda(n)}$ is an appropriate polynomial. Then, we express that ${F^{\lambda(n)}(Q^V)}$ is a singleton as the following formula, which intuitively says that for each node there is a state such that all orbits terminate in that state at this node: $\bigwedge_{s\in V}\bigvee_{q_0 \in Q} \mathcal{A} \models \mathcal{E}^{q_0,s}_{\lambda(n)},$
where $\mathcal{E}^{q_0,s}_{\lambda(n)}$ are $\lambda(n)$-specifications satisfying $\mathcal{E}^{q_0,s}_{\lambda(n)}(v) = Q^t,$ for every $v \not = s$ and $\mathcal{E}^{q_0,s}_{\lambda(n)}(s)$ is the set of orbits $y$ such that ${y_{\lambda(n)} = q_0}$.  The reduction holds.

It is straightforward to reduce coloring problems (does the graph admit a proper coloring with colors in $Q$) and more generally tilings problems to nilpotency using an error state that spread across the network when a local condition is not satisfied (note that tiling problem are known to be tightly related to nilpotency in cellular automata \cite{kari92}). Using the same idea one can reduce SAT to nilpotency by choosing $G$ to be the DAG of a circuit computing the given SAT formula (see Theorem~\ref{thm:nilpotency} below for a stronger reduction that works on any family of graphs with polynomial treewidth).

\begin{remark}
  \label{rem:nilpotencyontrees}
  If we allow the input automata network to be associated to a graph of unbounded degree (the local rule is then given as a circuit), it is possible to reduce any state instance to an automata network on a star graph with alphabet ${Q=\{0,1,\epsilon\}}$ where the central node simply checks that the Boolean values on leafs represent a satisfying instance of the SAT formula and produces a $\epsilon$ state that spreads over the network if it is not the case. The circuit representing the update rule of each node is NC in this case, and the automata network is nilpotent if and only if the formula is not satisfiable. 
\end{remark}

Given a deterministic freezing automata network of global rule ${F:Q^V\rightarrow Q^V}$, we define the associated non-deterministic global rule ${F^\ast}$ where each node can at each step to apply $F$ or to stay unchanged, formally: ${F^\ast_v(c) = \{F_v(c),c_v\}}$. It represents the system $F$ under totally asynchronous update mode.

\begin{problem}[Asynchronous reachability Problem]\ 
  \label{prob:reachability}
  \begin{description}
  \item[Parameters: ] alphabet $Q$, family of graphs $\mathcal{G}$ of max degree $\Delta$
  \item[Input: ]\ 
    \begin{enumerate}
    \item a deterministic freezing automata network
      $\mathcal{A}=(G,F)$ on alphabet $Q$, with set of nodes
      $V$ with ${n=|V|}$ and $G\in\mathcal{G}$;
    \item an initial configuration $c_0\in Q^V$
    \item a final configuration $c_1\in Q^V$
  \end{enumerate}
\item[Question: ] can ${c_1}$ is reached starting from $c_0$ under $F^\ast$?
  \end{description}
\end{problem}
Note that no bound is given in the problem for the time needed to reach the target configuration. However, Lemma~\ref{lem:trace-length} ensures that $c_1$ can be reached from $c_0$ if and only if it can be reach in a polynomial number of steps (in $n$). Thus this problem can again be seen as a sub-problem of our \textsc{SPEC} by defining a $\lambda(n)$-specification $\mathcal{E}_{\lambda(n)}$ such that for any $y \in \mathcal{E}_{\lambda(n)}:  y_0 = c_0 \wedge y_{\lambda(n)} = c_1$. This bound on the maximum time needed to reach the target ensures that the problem is NP (a witness of reachability is an orbit of polynomial length). Note that the problem is PSPACE-complete for general automata networks: in fact it is PSPACE-complete even when the networks considered are one-dimensional (network is a ring) cellular automata (same local rule everywhere) \cite{Dennunzio2017}.

\section{A fast-parallel algorithm for the Specification Checking Problem}

In this section we present a fast-parallel algorithm for solving the \textsf{Specification Checking Problem} when the input graph is restricted to the family of graphs with bounded degree and treewidth. More precisely, we show that the problem can be solved by a CREW PRAM that runs in the time $\mathcal{O}(\log^2(n))$ where $n$ is the amount of nodes of the network. Thus, restricted to graphs of bounded degree and bounded treewidth, \textsf{Specification Checking Problem} belongs to the class $\textbf{NC}$. 

To explain how our main algorithm  solve the latter problem, we will divide it in a number of sub-routines,  that can be executed efficiently in parallel.  Then, we will present an NC algorithm for  \textsf{Specification Checking problem} as a combination of this sub-routines. 
We begin fixing sets $Q$, $\mathcal{G}$,  and natural numbers $\Delta$ and $k$.  Let $\mathcal{A} = (G, \mathcal{F})$, $t$ and $\mathcal{E}_t$ be an instance of the \textsf{Specification Checking Problem}, that we consider for the following definitions.
\begin{definition}\label{def:locallyvalid} 
	A \emph{locally-valid trace of a node $v \in V$} is a function $\alpha: N[v]  \rightarrow Q^t$ such that:
	\begin{enumerate}
		\item $\alpha(v)_{s+1} \in F_v( (\alpha(u)_s)_{u \in N[v]})$ for all $0\leq s < t$,
		\item $\alpha(v)$ belongs to $\mathcal{E}_t(v)$.
	\end{enumerate} 
	We call the set of all locally-valid traces of $v$ as $LVT(v)$
\end{definition}

Roughly speaking, a locally-valid trace of a vertex $v$ is a sequence of state-transitions of all the vertices in $N[v]$ which are consistent with local rule of $v$, but not necessarily consistent with the local-rules of the vertices in $N(v)$.   We also ask that the state-transitions of $v$ satisfy the $(\{v\}, Q, t)$-specification $\mathcal{I}_v$.

Given two finite sets $A, B$, and a function $f: A \to B$. We define the restriction function of  $f$ to a subset $A' \subseteq A$ as the function $f|_A': A' \to B$ such that, for all $v \in A'$ we have that $f|_A'(v) = f(v)$.
\begin{definition}\label{def:partiallyvalid} 
	Let $U\subseteq V$ be a subset of nodes. A \emph{partially-valid trace of a set of nodes $U\subseteq V$} is a function $\beta: N[U] \rightarrow Q^t$ such that
 $\beta |_{N[v]}$ belongs to $LVT(v)$ for each $v \in U $.

	We call the set of all partially-valid traces of $U$ as $PVT(U)$
\end{definition}

Roughly, a partially-valid trace for a set $U$ is a sequence of state-transition of all the vertices in $N[U]$, which are consistent with the local rules of all vertices in $U$, but not necessarily consistent with the local-rules of the vertices in $N(U)$. 

Let $(W, F, \{X_w: w\in W\})$ be a rooted binary-tree-decomposition of graph $G$ with root $r$, that we assume that has width at most $(3\tw(G)+2)$. For $w\in W$, we  call $T_w$ the set of all the descendants of $w$, including $w$. 

Our algorithm consists in a dynamic programming scheme over the bags of the tree. First, we assume that $PVT(X_w)$ is nonempty for all bags $w \in W$, otherwise the answer of the \textsf{Specification Checking problem} is false. For each bag $w \in T$ and $\beta^w \in PVT(X_w)$ we call $\Sol_w(\beta^w)$ the partial answer of the problem on the vertices contained bags in $T_w$, when the locally-valid traces of the vertices in $X_w$ are induced by $\beta^w$.  We say that $\Sol_w(\beta^w) = \textbf{accept}$ when it is possible to extend $\beta^w$ into a partially-valid trace of all the vertices in bags of $T_w$, and \textbf{reject} otherwise. More precisely, if $w$ is a leaf of $T$, we define $\Sol_w(\beta^w) = \textbf{accept}$ for all $\beta^w \in PVT(X_w)$. For the other bags, $\Sol_w(\beta^w) = \textbf{accept}$ if and only if exists a $ \beta \in PVT(\bigcup_{z \in T_w} X_z)$ such that $\beta(u) = \beta^w(u)$, for all $u \in X_w$.  Observe that the instance of the \textsf{Specification Checking problem} is accepted when there exists a $\beta^r \in PVT(X_r)$ such that $\Sol_r(\beta^r) = \textbf{accept}$.
The following lemma is the core of our dynamic programming scheme: 

\begin{lemma}\label{lem:dynprog}
	Let $w$ be a bag of $T$ that is not a leaf and $\beta^w \in PVT(X_w)$. Then $\Sol_w(\beta^w) = \textbf{accept}$ and only if for each child $v$ of $w$ in $T_w$ there exists a  $\beta^{v} \in PVT(X_{v})$  such that
	\begin{enumerate}
		\item $\beta^w(u) = \beta^{v}(u)$ for all $u\in N[X_w] \cap N[X_{v}]$, 
		\item $\Sol_{v}(\beta^{v}) = \textbf{accept} $ 
	\end{enumerate}
	
\end{lemma}

\begin{proof}
	First, let us assume that $\Sol_w(\beta^w) = \textbf{True}$ and let $v$ be on of the children of $w$ in $T_w$. This implies that there exists a partially-valid trace $ \beta \in PVT(\bigcup_{z \in T_w} X_z)$ such that $\beta(u) = \beta^w(u)$, for all $u \in X_w$.  Observe  that $N[\cup_{z\in T_v} X_z]  \subseteq N[(\bigcup_{z \in T_w} X_z]$.   Since $\beta$ is defined over $N[(\bigcup_{z \in T_w} X_z]$, we can define $\beta^{v}$ and $\beta^{T_{v}}$ as the restrictions of $\beta$ to the sets $N[X_{v}]$ and $N[\cup_{z\in T_{v}} X_z]$, respectively. Observe that $\beta^{v}$ satisfies the condition (1) and (2) because, by definition, $\beta^{v}(u)$ and $\beta^w(u)$ are both equal to $\beta(u)$ for all $u\in N[X_w] \cap N[X_{v}]$. Moreover, $\Sol_{w_L}(\beta^{v}) = \textbf{accept}$ because $\beta^{T_v}$ is a partially-valid trace of $\bigcup_{z \in T_v} X_z$ such that $\beta^{v}(u)  = \beta(u) = \beta^{T_v}(u)$ for each $u \in X_v$.

	Conversely, suppose that we have that conditions (1), (2) for each child of $w$.  If $w$ is a leaf the proposition is trivially true. Suppose then that $w$ is not a leaf.  For each child $v$ of $w$, let  $\beta^v$ be the partially-valid trace of $X_v$ satisfying that $\Sol_v(\beta^v) = \textbf{accept}$ and $\beta^v(u) = \beta^w(u)$ for each $u\in N[X_w] \cap N[X_v]$. Since $\Sol_v(\beta^v) = \textbf{accept}$ we know that  $\beta^v$ can be extended into a partially-valid trace of  $\cup_{z\in T_v} X_z$, that we call $\beta^{T_v}$.   Let us call $v_1$ and $v_2$ the children of $w$. We define then the function $\beta: N[\cup_{z\in T_w} X_z] \rightarrow Q^t$. 
	
	$$\beta(u) =\left\{ \begin{array}{rcl} 
	\beta^{w} (u) & \textrm{ if } & u \in N[X_w]\\
	\beta^{T_{v_1}} (u) & \textrm{ if } &  u \in  N[\bigcup_{z \in T_{v_1}}X_z]  \\
	\beta^{T_{v_2}} (u) & \textrm{ if } &  u \in  N[\bigcup_{z \in T_{v_2}}X_z]\
	\end{array} \right. $$ 
	We claim that there is no ambiguity in the definition of $\beta$. First, we claim that $N[\bigcup_{z \in T_{v_1}}X_v] \cap N[\bigcup_{z \in T_{v_2}}X_v]$ is contained in $N[X_u]$. Indeed, let $u$ be a vertex in $N[\bigcup_{z \in T_{v_1}}X_z] \cap N[\bigcup_{z \in T_{v_2}}X_z]$. There are three possibilities: 
	\begin{itemize}
		\item $u$ belongs to a bag in $T_{v_1}$ and to another bag in $T_{v_2}$ . In this case necessarily $u \in X_w$, because otherwise the bags containing $u$ would not induce a (connected) subtree of $T$.
		\item $u$ is not contained in a bag of $T_{v_1}$. Since $u$ belongs to $N[\bigcup_{z \in T_{v_1}}X_z]$, there exists a vertex $\tilde{u}$  adjacent to $u$ and contained in a bag of $T_{v_1}$. Note that $X_w$ contains $\tilde{u}$, because otherwise all the bags containing $\tilde{u}$ would be in $T_{v_1}$. Then, no bag would contain both $u$ and $\tilde{u}$. That contradicts the property of a tree-decomposition that states that for each edge of the graph $G$, there must exist a bag containing both endpoints. We deduce $\tilde{u}$ is contained in $X_w$ and then $u$ is contained in $N[X_w]$.
		\item $u$ is not contained in a bag of $T_{v_2}$. This case is analogous to the previous one. 
		
	\end{itemize}
	Following an analogous argument, we deduce that $N[\bigcup_{z \in T_{v_1}}X_z]  \cap N[X_w]$ is contained in $N[X_{v_1}]$ and that $N[\bigcup_{z \in T_{v_2}}X_z]  \cap N[X_w]$ is contained in $N[X_{v_2}]$.  We deduce that $\beta$ is well defined. Moreover, $\beta$ is a partially-valid trace of $\bigcup_{z\in T_w} X_z$ which restricted to $N[X_w]$ equals $\beta^w$. We conclude that $\Sol_w(\beta^w) = \textbf{accept}$. 
\end{proof}

In order to solve our problem efficiently in parallel, we define a data structure that allows us efficiently encode locally-valid traces and partially-valid traces. More precisely, in $N[v]$ there are at most $|Q|^{\Delta}$ possible state transitions. Therefore, when $t$ is comparable to $n$, most of the time the vertices in $N[v]$ remain in the same state. Then, in order to efficiently encode a trace, it is enough to keep track only of the time-steps on which some state-transition occurs. 

Let $U$ be a set of vertices of $G$. A \emph{$(U,t)$-sequence} $\mathcal{S}$ is a function $\mathcal{S}: U \rightarrow Q^t$ such that the sequence $\mathcal{S}(u)$ is non-decreasing, for all $u\in U$.   For each $0 \leq s \leq t$ let us call $\mathcal{S}_s$ the sequence $(\mathcal{S}(u)_s)_{u\in U} \in |Q|^{|U|}$. Let  $\textsf{Times}(\mathcal{S}) = (t_0, t_1, \dots, t_\ell)$  be the strictly increasing sequence of minimum length satisfying that $\mathcal{S}_{t_i} = \mathcal{S}_{s}$ for each $t_i \leq s < t_{i+1}$ and each $0 \leq i < \ell$. Observe that $t_0 = 0$ and $\ell = \ell(\mathcal{S}) \leq |Q|^{|U|}$.
For a natural numbers $m$ and $\ell$, let us call $\langle m \rangle_\ell$ the binary representation of $m$ using $\ell$ bits, padded with $\ell - \lceil \log m \rceil$ zeros when $\ell > \lceil \log m\rceil $.

\begin{definition}\label{def:encoding}
	
	Let $\mathcal{S}$ be a  $(U,t)$-sequence.  A \emph{succinct representation of} $\mathcal{S}$, that we call $\epsilon(\mathcal{S})$ is a pair $(\Times(\mathcal{S}), \States(\mathcal{S}))$ such that:
	\begin{itemize}
		\item $\Times(\mathcal{S})$ is a list of elements of $\{0,1\}^{\lceil \log (t+1) \rceil}$ of length  $|Q|^{|U|}$, such that $$\Times(\mathcal{S})_i = \left\{  \begin{array}{ccc} \langle t_i \rangle_{\lceil \log (t+1) \rceil} & \textrm{ if } i \leq \ell(\mathcal{S})\\ \langle t \rangle_{\lceil \log (t+1) \rceil} & \textrm{ if } i > \ell(\mathcal{S}) \end{array}\right.$$
		
		\item $\States(\mathcal{S})$ is a matrix of elements of $\{0,1\}^{\lceil \log |Q| \rceil }$ of dimensions $|Q|^{|U|} \times  |U|$, such that, if we call $u_1, \dots, u_{|U|}$ the vertices of $U$ sorted by their labels, then:
		$$\States(\mathcal{S})_{i,j} = \left\{ \begin{array}{cl} \langle \mathcal{S}(u_j)_{t_i} \rangle_{\lceil \log |Q| \rceil} & \textrm{ if } i \leq \ell(\mathcal{S}) \\ \langle \mathcal{S}(u_j)_{t} \rangle_{\lceil \log |Q| \rceil}  & \textrm{ if } i > \ell(\mathcal{S})  \end{array}. \right. $$
		
	\end{itemize}
	We also call $\#(U,t) = |Q|^{|U|} \lceil \log (t+1) \rceil + |U| |Q|^{|U|} \lceil \log |Q| \rceil$
\end{definition}

Observe that $\epsilon(\mathcal{S})$ can be written using exactly  $N = \#(U,t)$ bits.  In other words, the succinct representations of  all $(U,t)$-sequences can be stored in the same number of bits, which is $\cO(|U|\log t)$.  Therefore, there are at most $2^{N}  = t^{g(|U|)}$ possible $(U,t)$-sequences, for some function $g$ exponential in $|U|$. Moreover, we identify the succinct representation of $(U,t)$-sequence $\mathcal{S}$  with a number $x \in \{0 \dots, 2^{N}\}$, such that $\epsilon(\mathcal{S}) = \langle x \rangle_{N}$.

The restriction of $\mathcal{E}_t$ to the nodes in $U$ is denoted $\mathcal{E}_t(U)$. When $U = \{u\}$ we denote $\mathcal{E}_t(\{u\})$ simply $\mathcal{E}_t(u)$.

\begin{definition} \label{lem:specificationencoding} Let $U$ be a set of vertices an let us call $N = \#(U,t)$. A \emph{succinct representation} of a $\mathcal{E}_t(U)$ is a Boolean vector $\mathcal{X} = \epsilon(\mathcal{E}_t(U))$ of length $2^N$ such that $\mathcal{X}_i = \textsf{True}$ when $i$ represents the succinct representation of a $(U,t)$-sequence contained in~$\mathcal{E}_t(U)$.

\end{definition} 

 Next lemma states that the succinct representation of a $(U,t)$-specification can be computed by fast parallel algorithm. 

\begin{lemma}
\label{lem:algospecification}
	For each set of vertices $U$, there exist a function $f$ and CREW PRAM algorithms performing the following tasks in time $f(|U||Q|)\log n$ using $n^{f(|U||Q|)}$ processors: 
	
	\begin{itemize}
		\item Given a $(U,t)$-sequence $\mathcal{S}$ as a $t \times |U|$ table of states in $Q$, compute $\epsilon(\mathcal{S})$
		\item Given a $\mathcal{E}_t(U)$ as a list of $(U,t)$-sequences, compute $\epsilon(\mathcal{E}_t(U))$
	\end{itemize}
\end{lemma}

\begin{proof}~
	\begin{itemize}
		\item
		The algorithm first computes $\textsf{Times}(\mathcal{S})$. Then, it constructs the list $\Times(\mathcal{S})$ and the matrix $\States(\mathcal{S})$ copying the lines of $\mathcal{S}$ given in $\textsf{Times}(\mathcal{S})$. 
		
		The algorithm starts reserving $N = \#(U,t)$ bits of memory for in the list $\Times$ and the matrix $\States$, and $t+1$ bits of memory represented in a vector $\textsc{indices}$. The vector $\textsc{indices}_s$ sores the time-steps on which that belong to $\textsf{Times}(\mathcal{S})$. 
		
		For each $i \in \{1, \dots, t\}$ the algorithm initializes a processor $P_i$ and assigns the $i$-th bit of $\textsc{indices}$ to it. Processor $P_i$ looks at the $i$-th and $i-1$-th lines of $\mathcal{S}$. If $\mathcal{S}_{i} \neq \mathcal{S}_{i-1}$  then processor writes a $1$ in $\textsc{indices}_i$. Otherwise, the processor writes a $0$ in $\textsc{indices}_i$. Then $P_i$ stops. All this process can be done in time $\cO(|U|\log|Q| + \log t)$ per processor. 
		
		Then, the algorithm computes the vector $p$ of length $t$ such that $p_j = \sum_{j=1}^i \textsc{indices}_j$, for each $j\in \{1, \dots, t\}$. This process can be done in time $\cO(\log t)$ using $\cO(t)$ processors using the prefix sum algorithm given by \cite{jaja} (Proposition \ref{prop:prefixsum}). Observe that if $ \textsc{indices}_i = 1$ for some index $i$, then $i = \textsf{Times}(\mathcal{S})_{p_i}$. Moreover, $p_t = \ell(\mathcal{S})$.
		
		Once every processor $(P_i)_{0< i\leq t}$ stops, the algorithms reinitialize them. For each $0<i\leq t$, each processor $P_i$ looks at $\textsc{indices}_i$. If $p_i < p_t$ and $\textsc{indices}_i= 0$ then processor $P_i$ stops. If  $p_i = p_{i-1} = p_t$ the processor stops. If $p_i \neq p_t$ and $\textsc{indices}_i= 1$, then the algorithm writes $\langle i \rangle_{\lceil \log (t+1) \rceil}$ in $\Times_{p_i}$, and for each $u \in \{1, \dots, |U|\}$ writes $\langle \mathcal{S}_{i,u}\rangle_{\lceil \log |Q| \rceil}$ in $\States_{p_i,u}$. If $p_i = p_t$ and $p_{i-1} \neq p_i$, then the processor $P_i$ writes $\langle t \rangle_{\lceil \log (t+1) \rceil}$ in $\Times_{j}$ and writes $\langle \mathcal{S}_{t,u}\rangle_{\lceil \log |Q| \rceil}$ in $\States_{j,u}$ for each $p_i \leq j \leq |Q|^{|U|}$  and for each $u \in \{1, \dots, |U|\}$. The algorithm writes $\Times_0 = \langle 0 \rangle_{\lceil \log (t+1) \rceil}$  and writes $\langle \mathcal{S}_{0,u}\rangle_{\lceil \log |Q| \rceil}$ in $\States_{j,u}$ for each $u\in \{1, \dots, |U|\}$.  All this process can be done in time $\cO(|Q|^{|U|} \log t)$ per processor. 
		
		The algorithm returns $\epsilon(\mathcal{S}) = (\Times, \States)$. The whole process takes time  $\cO(|Q|^{|U|} \log t)$ and $\cO(t)$ processors. 
		\item The algorithm initializes $\mathcal{X} =\epsilon(\mathcal{E}_t(U))$ as $2^N$ bits of memory bits, all in $0$. Then, it assigns one processor $P_{\mathcal{S}}$ to each $(U, t)$-sequence $\mathcal{S}$ in $\mathcal{E}_t(U)$. For each $\mathcal{S}\in \mathcal{E}_t(U)$, processor  $P_{\mathcal{S}}$  uses the previous algorithm to compute $y = \epsilon(\mathcal{S})$. Then processor $P_{\mathcal{S}}$ writes $\mathcal{X}_y = 1$. 
		Once every processor has finished, the algorithm returns $\mathcal{X}$. The whole process takes time  $\cO(\log|\mathcal{X}| |Q|^{|U|} \log t)$ and uses $|\mathcal{X}|t^{\cO(1)} = n^{\cO(|Q|^{|U|})}$ processors. We deduce that the algorithm runs in time $\cO(|Q|^{|U|}\log t)$ using  $2^N$  processors. \end{itemize}
\end{proof}

Observe that if  $\beta$ is a partially-valid trace of $U$, then in particular  $\beta$ is a $(N[U], t)$-sequence. Therefore, there exists an $x \leq 2^N$ with $N = \#(N[U],t)$,  such that $\epsilon(\beta) = x$.  In the following lemma we show how to characterize the values on  $x \leq 2^N$ that are the encoding of some partially-valid trace of $U$. We need the following definition.
Let $U$ be a set of vertices and let $x\in \{0, \dots, 2^N\}$, with $N= \#(U,t)$. Then we call $\Times(x)$ and $\States(x)$ the vector and matrix such that $x = (\Times(x), \States(x))$. More precisely:
\begin{itemize}
	\item  $\Times(x)$ are the first $|Q|^{|U|}$ bits of $x$ interpreted as sequence of elements of $\{0,1\}^{\lceil \log(t+1)\rceil}$ of length $|Q|^{|U|}$. 
	\item $\States(x)$ are the rest of the bits of $x$ interpreted as the matrix of elements of $\{0,1\}^{\lceil \log |Q|\rceil}$ of dimensions $|Q|^{|U|} \times |U|$. 
\end{itemize}

\begin{lemma}
\label{lem:succintSubsets} Let $S$ be a $(U,t)$-sequence and $Z\subseteq U$. 
	There is a sequential algorithm which given $\epsilon(\mathcal{S})$  computes $\epsilon(\mathcal{S}|_Z)$ in time linear in the size of $\epsilon(\mathcal{S})$.
\end{lemma}

\begin{proof}
	
	Let $\kappa = |Q|$. The  computes algorithm  $\epsilon(\mathcal{S}|_Z)$ checking each pair of lines of $\States$ and verifying if the columns of $Z$ differ on any coordinate, keeping only the lines on which some of the vertices in $Z$ switches states for the first time. More precisely, let $u_1, \dots, u_\kappa$ be the set $U$ ordered by their labels.  Let $J \in \{j_1, \dots, j_{|Z|}\}$ be the set of indices of vertices of $Z$ (i.e., $u_{j_q} \in Z$ for all $q \in \{1, \dots, |Z|\}$). The algorithm computes the set $L$ of indices  $i \leq |Q|^\kappa$ such that $i\in L$ if and only if there exists $q \in J$ such that $\States_{i,j_q} \neq \States_{i-1,j_q}$.  Let $\{i_1, \dots, i_{|L|}\}$ the indices in $L$. Observe that $|L| \leq |Q|^{|Z|}$. Then for each $p \leq |Q|^{|Z|}$ and $q \in |Z|$,
	
	$$\Times[Z]_p = \left\{  \begin{array}{ccc} \Times_{i_p} & \textrm{ if } p \leq |L| \\ \langle t \rangle_{\lceil \log (t+1) \rceil} & \textrm{ if } i > |L| \end{array}\right.$$
	
	$$\States[Z]_{p,q} = \left\{ \begin{array}{cl} \States_{i_p, j_q} & \textrm{ if } p \leq |L|  \\ \States_{t, j_q} & \textrm{ if } p > |L|    \end{array}. \right. $$
	
	The algorithm returns $(\Times[Z], \States[Z])$. 
\end{proof}

\begin{lemma}
	 \label{lem:PVTverif}
	Let $U$ be a set of vertices and let $N = \#(N[U],t)$.  There is a sequential algorithm which, given $x \geq 0 $ and  $\epsilon(\mathcal{E}_t(u))$ for each $u\in V(G)$,   decides  in time $f(|N[U]||Q|)\log n$ whether  $x$ is a succinct representation of a partially-valid trace of $U$, where $f$ is an exponential function.
\end{lemma}

\begin{proof}
	
	Let $U$ be a set of vertices containing $S$ and let $x\in \{0, \dots, 2^N\}$, with $N= \#(N[U],t)$. Let $\kappa = |N[U]|$. Let $\{u_1, \dots, u_{\kappa}\}$ be the vertices of $N[U]$ ordered by label. The algorithm first verifies that $x \leq 2^N$ and rejects otherwise.  Then, the algorithm verifies that the pair $(\Times, \States)= (\Times(x), \States(x))$ satisfies $\Times_0 = 0$ and that $\Times$ and each column of $\States$ are increasing. Otherwise, the algorithm rejects because $x$ is not a succinct representation of a $(N[U], t)$-sequence. If the algorithm passes this test we assume that $x= \epsilon(\mathcal{S})$ for some $(N[U],t)$-sequence $\mathcal{S}$.  For a subset of vertices $Z$, let us call $(\Times[Z], \States[Z]) = \epsilon(\mathcal{S}|_Z)$.
	Consider now the following conditions:

	\begin{enumerate}
		\item $\States_{i,j} \in F_{u_j}( \States[N[u_j])]_i)$ for each $i \in \{0, \dots, |Q|^{\kappa}\}$ and $j \in \{1, \dots, \kappa \}$ such that $u_j \in U$.
		\item $(\Times[\{u_j\}], \States[\{u_j\}])$ belongs to $\mathcal{E}_t(u_j)$ for each $j \in \{1, \dots, \kappa \}$.
	\end{enumerate}
	
	When this conditions are satisfied, we can deduce that $x = \epsilon(\beta)$ for some partially-valid trace $\beta$ of $U$. Indeed, as $x$ is  representation of $\mathcal{S}$,  the vertices in $N[U]$ only have state-transitions of the time-steps given by the $\Times$.  Therefore,  condition (1.) and (2.) imply that $\mathcal{S}|_{N[u]}$ is a locally-valid trace of $u$. 
	To verify condition (1.) and (2.) we use the algorithm of Lemma \ref{lem:succintSubsets} to  compute $\States[Z]$ for a given set of vertices $Z\subseteq N[U]$. Observe that the algorithm computes $\Times[Z]$ and $\States[Z]$ in time $\cO(\kappa |Q|^\kappa \log n)$. 
	The algorithm checks (1.) by looking at each row of  $\States[N[u]]$ and the column corresponding to vertex $u$, and the table of $F_u$ given in the input.  The algorithm verifies (2.)  computing  $\epsilon(\mathcal{S}(u_j)) = (\Times[\{u_j\}], \States[\{u_j\}])$ and then looking at the  $\epsilon(\mathcal{S}(u_j))$-element of the table $\epsilon(\mathcal{E}_t(u_j))$. All these processes take time   $\cO(|U|\kappa |Q|^\kappa \log n)$.
	Overall the whole algorithm takes time $\cO(|U|\kappa |Q|^\kappa \log n) = f(|N[U]||Q|) \log n$.
\end{proof}

We are now ready to give our algorithm solving the \textsf{Specification Checking problem}.
\begin{theorem}\label{lem:algospeci}
	\textsf{Specification Checking problem} can be solved by an CREW PRAM algorithm running in time $\cO(\log^2 n)$ and using $n^{\mathcal{O}(1)}$ processors on graphs of bounded treewidth. 
\end{theorem} 

\begin{proof}
	
	Our algorithm consists in an implementation of the dynamic programming scheme explained at the beginning of this section. Our algorithm starts computing a rooted binary-tree decomposition $(W, F, \{X_w: w\in W\}$ of the input graph using the logarithmic-space algorithm given by Proposition \ref{prop:treedecompose}. 
	The algorithm also computes the succinct representations of $\mathcal{E}$ and $\mathcal{I}_v$ for each $v\in V$ using Lemma~\ref{lem:algospecification}.
	

	Then, the algorithm preforms the dynamic programming scheme over $T$.  Let $r$ be the root of $T$. The for a bag $w\in W$, we define the \emph{level of $w$} denoted by $\textsc{Level}(w)$, as the distance between $w$ and the root $r$. There is a fast-parallel algorithm computing the level of each vertex of a tree by a EREW PRAM running in time $\cO(\log n)$ and using $\cO(n)$ processors \cite{jaja}. Using a prefix-sum algorithm we can compute the maximum level $M$ of a vertex, which correspond to the leafs of the binary-tree $T$. For each $i \in \{0, \dots, M\}$, let $\mathcal{L}_i$ the set of bags $w$ such that $\Level(w) = M-i$.

	For each $w\in W$, we represent the values of the function $\Sol_w$ as a table $S^w$ indexed as a table of size $2^N$, with $N= \cO(|Q|^{\Delta(3\tw(G)+2+k)}\log n)$ greater that $\#(N[X_w],t)$ for all $w \in W$. Each $x \in \{0, \dots, 2^N\}$ is interpreted as a potentially succinct encoding of a  partial-valid trace $\beta$. Initially $S^w = 0^{2^N}$, which meaning that a priori we reject all  $x \in \{0, \dots, 2^N\}$. Then, our algorithm iterates in a reverse order over the levels of the tree, starting from $\mathcal{L}_0$ until reaching the the root $r \in \mathcal{L}_M$. In the $i$-th iteration, we compute for each  bag $w\in \mathcal{L}_i$ the set of all  $x \in \{0, \dots, 2^N\}$ that represent partially-valid traces $\beta^w \in PVT(X_w)$ such that $\Sol_w(\beta^w) = \textbf{accept}$. To do so, the algorithm uses the calculations done on the bags in $\mathcal{L}_{i-1}$, and use Lemma \ref{lem:dynprog}. The algorithm saves the answer of each partial solution in a variable $\textsf{out}$ consisting in $|W|$ bits, such that, and the end of the algorithm $\textsf{out} = 1^{|W|}$ if and only the instance of the \textsf{Specification Checking problem} is accepted.
	
	At the first iteration, for each $w\in \mathcal{L}_0$ the algorithm sets in parallel $S^w_x = 1$ for all $x$ representing a partially-valid trace of $w$, because $\Sol_w(\beta^w)$ is defined to  $\textbf{accept}$ for all partially-valid trace of a leaf of $T$.  Therefore, in parallel for all bag $w \in  \mathcal{L}_0$, the algorithm runs $2^N$ parallel instances of the algorithm of Lemma \ref{lem:PVTverif}, one for each $x\in \{0, \dots, 2^N\}$, and for each one that is accepted, the algorithm writes $S^w_x = 1$. Once every parallel verification finishes, the algorithm sets $\textsf{out}_w=1$. We now detail the algorithm on the $i$-th iteration, assuming that we have computed $S^w$ for all bag $w\in  \mathcal{L}_{i-1}$.  
	
	Let $w$ be a vertex in $\mathcal{L}_i$ and let us call $w_L$ and $w_R$ the children of $w$, which belong to $\mathcal{L}_{i-1}$. Roughly, as we know the partial solutions restricted to the subtrees rooted at $w_1$ and $w_2$, the algorithm will try to extend it to a partial solution of $w$ according to the gluing procedure given by Lemma  \ref{lem:dynprog}, testing all possible combinations. More precisely, we initialize a set $|\mathcal{L}_i|$ processors $\{P^w\}_{w\in \mathcal{L}_i}$, one assigned each bag in $\mathcal{L}_i$. Each processor $P^w$ verifies if $\textsf{out}_{w_L} = \textsf{out}_{w_R} = 1$, or stops and writes $\textsf{out}_{w} = 0$. Otherwise, processor $P^w$ initializes a set of  $2^N$ processors, that we call $\{ P^w_{z}\}_{z  \in \{1, \dots, 2^N\}}$, and reserves $2^N$ bits of memory $S^w\in \{0,1\}^{2^N}$. For each $z\in 2^N$, processor $P^w_{z}$ verifies if $z$ is a succinct representation of a partially-valid trace of $X_w$ using Lemma~\ref{lem:PVTverif}. If its not the case then $P^w_{z}$ stops and writes a $0$ in $S^w_z$.  Otherwise, processor $P^w_{z}$ initializes $(2^N)^2$ processors  $\{ P^w_{z, z_R, z_L} : z_R, z_L  \in \{1, \dots, 2^N\} \}$ and reserves $2^N \times 2^N$ bits of memory $(\tilde{S}^w_{z})\in \{0,1\}^{N} \times  \{0,1\}^{N}$. 
	
	If $S^{w_L} _{z_L}=0$ or $S^{w_R} _{z_R}=0$ the processor $P^w_{z, z_R, z_L}$  stops and writes a $0$ in $\tilde{S}^w_{z,z_R,z_L}$. Otherwise, the processor $P^w_{z, z_R, z_L}$ interprets $z, z_R$ and $z_L$ as  $\epsilon(\beta^w_z)$, $\epsilon(\beta^w_{z_L})$ and $\epsilon(\beta^w_{z_R})$, for partially-valid traces $\beta_z$, $\beta_{w_L}$ and $\beta_{w_R}$ of $X_w$, $X_{w_L}$ and $X_{w_R}$, respectively.  Which means that $\beta^w_z$ belongs to $PVT(X_w)$ and $\Sol_{w_L}(\beta_{z_L}) = \Sol_{w_L}(\beta_{z_L}) = \textbf{accept}$. Therefore $\beta_z$ is a partially-valid trace of $X_w$ and $\beta_{z_L}$ and $\beta_{z_R}$ verify the condition (2) of Lemma~\ref{lem:dynprog}. Up to this point, all verifications can be done in time $\cO(N) = \cO(|Q|^{\Delta(3\tw(G) +2 +k)}\log n)$ because we are just looking at the coordinates in the given tables. 
	
	Then, the processor $P^w_{z, z_L, z_R}$ computes sets $Y_L = N[X_w] \cap N[X_{w_L}]$ and using the algorithm of Lemma~\ref{lem:succintSubsets} computes $\epsilon(\beta^w|_{Y_L})$ and $\epsilon(\beta^{w_L}|_{Y_L})$. If $\epsilon(\beta^w|_{Y_L}) = \epsilon(\beta^{w_L}|_{Y_L})$ the processor deduces that $\beta^w(u) = \beta^{w_L}(u)$,  for all $u\in N[X_w] \cap N[X_{w_L}]$. Then $P^w_{z, z_L, z_R}$ computes sets $Y_R = N[X_w] \cap N[X_{w_R}]$ and using the algorithm of Lemma~\ref{lem:succintSubsets} computes $\epsilon(\beta^w|_{Y_R})$ and $\epsilon(\beta^{w_R}|_{Y_R})$. Then, if $\epsilon(\beta^w|_{Y_R}) = \epsilon(\beta^{w_R}|_{Y_R})$ the processor deduces that $\beta^w(u) = \beta^{w_R}(u)$,  for all $u\in N[X_w] \cap N[X_{w_R}]$. If both verifications are satisfied, processor $P_{z, z_R, z_L}^w$ stops and writes a $1$ in $\tilde{S}^w_{z,z_R,z_L}$.  Otherwise, the processor $P_{z, z_R, z_L}^w$ stops and writes a $0$ in $\tilde{S}^w_{z,z_R,z_L}$. All of these verifications can be executed by $P^w_{z, z_L, z_R}$ in time $\cO(N)$.

	Once that all processors in $\{P_{z, z_1, z_2}^w : z_L, z_R \in \{1, \dots 2^N\}\}$ finished, processor $P^w_z$ runs a prefix-sum algorithm in $\tilde{S}^w_z$, simply summing the elements of the vector to verify if some instance was accepted. If the result is different than $0$, processor $P^w_z$ writes a $1$ in $S^w_z$, and writes a $0$ otherwise. 
	When every processor $(P^w_z)_{z\in \{1, \dots, 2^N\}}$ finishes, we obtain that $S^w$ is the table representing function $\Sol_w$. Then processor $P^w$ runs a prefix-sum algorithm on $S^w$ to verify that there exists a partial solution for bag $X_w$. If the result of the prefix sum equals zero, processor $P^w$ stops and writes a  $\textsf{out}_w = 0$. Otherwise, it writes $\textsf{out}_w = 1$ and stops. 
	
	After all processors $(P^w)_{w\in \mathcal{L}_i}$ have finished, the algorithm continues with the next level. When the last level is reached, before halting processor $P^r$ decides if $\textsf{out}_w = 1$ for all $w\in W$ using a prefix-sum algorithm. If the answer is affirmative the algorithm \textbf{accepts} the input, and otherwise \textbf{rejects}. On each level, the algorithm takes time $\cO(\Delta|Q|^{2\Delta(3\tw(G)+2)}\log n)$ and uses $n^{\cO(|Q|^{\Delta (3\tw(G)+2)})}$ processors. Proposition \ref{prop:treedecompose} provides a construction of a binary-tree-decomposition $T$ of depth $\cO(\log n)$. This means that $M = \cO(\log n)$, and implies that the whole takes time  $\cO(\Delta|Q|^{2\Delta(3\tw(G)+2)}\log^2 n) = \cO(\log^2 n)$ and $n^{\cO(|Q|^{\Delta (3\tw(G)+2)})} = n^{\cO(1)}$ processors. The correctness of the algorithm is given by Lemmas  \ref{lem:dynprog}, \ref{lem:algospecification}, \ref{lem:succintSubsets} and \ref{lem:PVTverif}.
\end{proof}

\begin{remark}\label{rem:esp}
	The algorithm given in the proof of Theorem \ref{lem:algospeci} not only computes the answer of \textsf{Specification Checking problem} but it also gives the coding of the orbits satisfying specification $\mathcal{E}_t$.
	\label{rem:computingtraces}
\end{remark}
\begin{remark}\label{rem:det}
	In the case in which the freezing automata network $\mathcal{A} = (G, \mathcal{F})$ is deterministic, we can say a lot more using latter algorithms. Giving $t$ and an initial condition $x \in Q^n$,  we are actually capable of testing any global dynamic property in  \textbf{NC} provided that this property has $F^t(x)$ as input and it is decidable in $\textbf{NC}$. In fact,  note that given an initial condition $x \in Q^n$,  there is only one possible orbit for each node $v \in V(G)$. Therefore, as a consequence of Remark \ref{rem:computingtraces} we are able to calculate the global evolution of the system in time $t$ starting from $x$. 
\end{remark}

The proof of previous Theorem \ref{lem:algospeci}  shows that \textbf{SPEC} can be solved in time $f(|Q|+\Delta(G)+ \tw(G))\log n$ using $n^{f(|Q|+\Delta(G)+ \tw(G))}$ processors in a PRAM machine, hence in time $n^{g(|Q|+\Delta(G)+ \tw(G))}$ on a sequential machine, for some computable functions $f$ and $g$. In other words, when the alphabet, the maximum degree and the tree-width of the input automata network are parameters, our result shows that \textbf{SPEC} is in XP. In the next section, we show that \textbf{SPEC} is not in FPT, unless FPT=W[2].

\subsection{Constraint Satisfaction Problem} 

We remark that problem $\textsc{SPEC}$ can be interpreted as a specific instance of the Constraint Satisfaction Problem (CSP). Problem CSP is a sort of generalization of SAT into a set of more versatile variable constraint. It is formally defined as a triple $(X, D, C)$, where $X = \{X_1, \dots, X_n\}$ is a set of \emph{variables}, $D = \{D_1, \dots, D_n\}$ is a set of \emph{domains} where are picked each variable, and a set $C = \{C_1, \dots, C_m\}$ of \emph{constraints}, which are $k$-ary relations of some set of $k$ variables. The question is whether exists a set of values of each variables in their corresponding domains, in order to satisfy each one of the constraints.  As we mentioned, $\textsc{SPEC}$ can be seen as a particular instance of CSP, where we choose one variable for each node of the input graph. The domain of each variable is the set of all locally valid traces of the corresponding node. Finally, we define one constraint for each node, where the variable involved are all the vertices in the close-neighborhood of the corresponding node, and the relation corresponds to the consistency in the information of the locally-valid traces involved. 

Now consider an instance of $\textsc{SPEC}$ with constant tree-width, maximum degree and size of the alphabet, and construct the instance of CSP with the reduction described in the previous paragraph. Then, the obtained instance of CSP has polynomially-bounded domains and constant tree-width, where the tree-width of a CSP instance is defined as the tree-width of the graph where each variable is a node, and two nodes are adjacent if the corresponding variables appear in some restriction. Interestingly, it is already known that in these conditions CSP can be solved in polynomial time  \cite{samer2010constraint,dalmau2002constraint,marx2007can}  . This implies that, subject to the given restrictions, \textsc{SPEC} is solvable in polynomial time using the given algorithm for CSP as a blackbox. 

The algorithm given in the proof of Theorem \ref{lem:algospeci} is better than the use of the CSP blackbox in two senses. First, we obtain explicit dependencies on the size of the alphabet, maximum degree and tree-width. Second, the \emph{Prediction Problem} is trivially solvable in polynomial time, and then the use of the CSP blackbox gives no new information for this problem. Moreover, as we mentioned in Remarks \ref{rem:esp} and \ref{rem:det}, our algorithm does not decides $\textsc{SPEC}$, but  also can be used to obtain a coding of the orbit satisfying the given specification, and moreover, the possibility to test any NC-property on deterministic freezing automata networks.  

\section{$W[2]$-hardness results}
\label{sec:w2}

The goal of this section is to show that, even when alphabet and degree are fixed and treewidth is considered as the only parameter, then the \textbf{SPEC} problem is $W[2]$-hard (see \cite{Downey_2013} for an introduction to the $W$ hierarchy) and thus not believed to be fixed parameter tractable. This is in contrast with classical results of Courcelle establishing that model-checking of MSO formulas parametrized by treewidth is fixed-parameter tractable \cite{Courcelle_1990} (see \cite{beyondmso,Downey_2013} to place the result in a wider context).

\begin{lemma}
	\label{lem:kdominating}
	There is a fixed alphabet $Q$ and an algorithm which, given ${k\in\N}$ and a graph $G$ of size $n$, produces in time ${O(k\cdot n^{O(1)})}$:
	\begin{itemize}
		\item a deterministic freezing automata network
		$\mathcal{A} = (G', \mathcal{F})$ with alphabet $Q$ and where $G'$ has treewidth ${O(k)}$ and degree $4$
		\item a $O(n^2)$-specification ${\mathcal{E}}$
	\end{itemize}
	such that $G$ admits a dominating set of size $k$ if and only if $\mathcal{A} \models \mathcal{E}$.
\end{lemma}
The construction of the lemma works by producing a freezing automata network on a ${O(k)\times n^2}$-grid together with a specification which intuitively work as follows. A row of the grid is forced (by the specification) to contain the adjacency matrix of the graph, $k$ rows serve as selection of a subset of $k$ nodes of $G$, and another row is used to check domination of the candidate subset. The key of the construction is to use the dynamics of the network to test that the information in each row is encoded coherently as intended, and raise an error if not. The specification serves both as a partial initialization (graph adjacency matrix and tests launching are forced, but the choice in selection rows is free) and a check that no error are raised by the tests.
\begin{proof}
	Let ${G'=(V',E')}$ be the ${(k+2)\times n^2}$-grid where
	${V'=\{(i,j) : 0\leq i<n^2, 1\leq j\leq k+2\}}$ and
	\begin{align*}
	E' &= \{\{(i,j),(i\pm 1\bmod n^2,j)\} : 0\leq i<n^2, 1\leq j\leq k+2\}\\
	&\cup \{\{(i,j),(i,j')\} : 0\leq i<n^2, 1\leq j,j'\leq k+2, |j-j'|=1\}.
	\end{align*}
	Clearly $G'$ has a ${O(k)}$ treewidth.
	The horizontal dimension (coordinate $i$ in the grid) should be thought as $n$ block of size $n$.
	For each $j$ we denote the $j$th row by ${V_j = \{(i,j) : 0\leq i<n^2\}}$.
	The alphabet of the automata network is ${Q=\{0,1\}\times Q'}$ where the ${\{0,1\}}$ is the \emph{marker component} and $Q'$ is \emph{verification component}.
	A position is said to be marked if it has a $1$ in its marker component. 
	Vertically, the network is organized as follows.
	\begin{itemize}
		\item Rows $V_1$ to $V_k$ are called \emph{selection rows} and they
		all have the same behavior: marking the same unique position in
		each block, \textit{i.e.} having horizontal coordinates $s$,
		${s+n}$, ${s+2n}$, ${\ldots}$, ${s+(n-1)n}$ marked for some $s$
		with ${0\leq s < n}$. Intuitively the role of each selection row is to select a node among the $k$ nodes of the candidate dominating set and ensure that the selection information is coherently spread across the $n$ blocs.
		\item Row $V_{k+2}$ is the \emph{graph row} and its role is to hold the adjacency matrix of graph $G$ laid out in a single row (bloc $i$ contains the incidence vector of node $i$).
		\item Row $V_{k+1}$ is the \emph{domination row}. Its role is to witness that there is a position $i$ in each bloc (possibly different from one bloc to another) where $i$ is marked in the graph row and also marked in at least one selection row. Said differently its role is to give a certificate that the $k$ selected nodes in the selection rows are indeed a dominating set for the graph encoded in the graph row.
	\end{itemize}
	The description by rows above gives some conditions on the marked position in the network. It should be clear that these conditions are satisfied if and only if the $k$ selection rows represent $k$ nodes of graph $G$ that form a dominating set.
	
	We now complete the description of ${\mathcal{A}}$.
	In the verification component of states $Q'$ there is a special error state. The behavior of the automata network is to perform two tests to ensure that each row is holding marks that satisfies the conditions above. If any test fails somewhere the error state is raised and stays forever. The two test run in parallel (using two independent subcomponents inside $Q'$) and their technical implementation as a freezing automata network on graph $G'$ is straightforward using a constant number of state component within $Q'$. Note that in the description below, what we call signals are freezing signals: a state change from ${q}$ to ${q''}$ possibly with intermediate state $q'$ with ${q\leq q'\leq q''}$ that propagates in some direction like a flame in a wick (and not a particle in state $q_1$ that move inside a context of $q_0$ like classical signals are). The tests are as follows:
	\begin{itemize}
		\item the domination test works vertically: each marked position $i$ in the domination row checks that the position $i$ is also marked in the neighboring graph row and then launch a signal that moves from ${(i,k+1)}$ downto ${(i,1)}$ until it finds a marked position in some selection row. If the signal reaches position ${(i,1)}$ without having encountered any mark, then it raises an error state.
		\item the selection test happens in each selection row horizontally: first it checks that exactly one position is marked within each bloc (this can be done in one step by using a layer of alphabet ${\{>,<\}}$ and forcing the language ${>^\ast<^+}$ by forbidding the pattern ${<>}$ to appear in two adjacent position in a bloc); second, each marked position (in each bloc) launches a signal going left and a signal going right both moving at the speed of one position per time step. Each signal goes one, crosses a first signal going in the opposite direction, continues, and when it encounters a signal of the opposite direction for the second time it stops and checks that the position reach contains a mark. This process ensures that the position marked is the same in each bloc by comparing the distance between marks in bloc ${p-1}$ and bloc ${p+1}$ and the distance between marks in blocs ${p}$ and ${p+1}$ for all $p$ (see Figure~\ref{fig:markdist}).
	\end{itemize}
	\begin{figure}
          \centering
          \begin{tikzpicture}[scale=.4]
            \draw[->] (-.5,0) -- node[sloped,above,midway] {time} +(0,5);
            \fill[color=gray!10!white] (1,0)--(0,1)--(0,5)--(3,2)--cycle;
            \fill[color=gray!50!white] (0,5)--(3,2)--(5,4)--(7,2)--(10,5)--cycle;
            \fill[color=gray!10!white] (3,2)--(5,0)--(7,2)--(5,4)--cycle;
            \fill[color=gray!10!white] (9,0)--(7,2)--(10,5)--(12,5)--(12,3)--cycle;
            \draw (5,4) node[circle,fill=green,minimum size=1mm,inner sep=2pt] { };
            \draw[thick,green] (5,4)-- +(0,1);
            \draw (5.5,4) node[right] {\tiny ok};
            \draw (0,0)--(12,0);
            \draw (0,-.2) -- (0,.2);
            \draw (2,0) node[below] {\tiny bloc $p-1$};
            \draw (4,-.2) -- (4,.2);
            \draw (6,0) node[below] {\tiny bloc $p$};
            \draw (8,-.2) -- (8,.2);
            \draw (10,0) node[below] {\tiny bloc $p+1$};
            \draw (12,-.2) -- (12,.2);
            \draw[dotted] (1,0) -- +(0,5);
            \draw[dotted] (5,0) -- +(0,5);
            \draw[dotted] (9,0) -- +(0,5);
          \end{tikzpicture}\hskip 1cm
          \begin{tikzpicture}[scale=.4]
            \fill[color=gray!10!white] (2,0)--(0,2)--(0,5)--(3.5,1.5)--cycle;
            \fill[color=gray!50!white] (0,5)--(3.5,1.5)--(5.5,3.5)--(7,2)--(10,5)--cycle;
            \fill[color=gray!10!white] (3.5,1.5)--(5,0)--(7,2)--(5.5,3.5)--cycle;
            \fill[color=gray!10!white] (9,0)--(7,2)--(10,5)--(12,5)--(12,3)--cycle;
            \draw (5.5,3.5) node[circle,fill=red,minimum size=1mm,inner sep=2pt] { };
            \draw[thick,red] (5.5,3.5)-- +(0,1.5);
            \draw (6,3.5) node[right] {\tiny error};
            \draw (0,0)--(12,0);
            \draw (0,-.2) -- (0,.2);
            \draw (2,0) node[below] {\tiny bloc $p-1$};
            \draw (4,-.2) -- (4,.2);
            \draw (6,0) node[below] {\tiny bloc $p$};
            \draw (8,-.2) -- (8,.2);
            \draw (10,0) node[below] {\tiny bloc $p+1$};
            \draw (12,-.2) -- (12,.2);
            \draw[dotted] (2,0) -- +(0,5);
            \draw[dotted] (5,0) -- +(0,5);
            \draw[dotted] (9,0) -- +(0,5);
          \end{tikzpicture}
          \caption{Checking that the same node is marked in each bloc in a selection row: on the left, a valid test in bloc $p$, on the right an invalid test in bloc $p$ generating an error state. Dotted lines indicate the marked position in each bloc. The shades of gray indicates the state changes involved in the implementation of freezing signals: at each position the sequence of states in non-decreasing with time.}
          \label{fig:markdist}
	\end{figure}
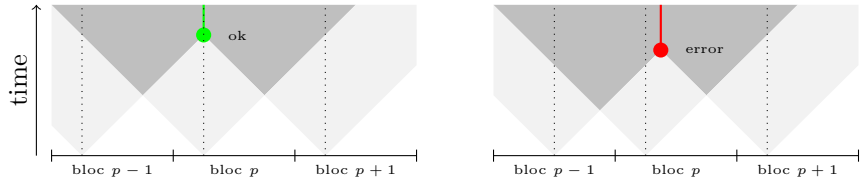
	
	Finally, the ${n}$-specification ${\mathcal{E}}$ consists in:
	\begin{itemize}
		\item forcing the initial marking of the graph row to be the actual adjacency matrix of $G$;
		\item allowing any marking in the other rows;
		\item forcing the $Q'$ component to be without error at any time step;
		\item initializing the $Q'$ component to properly launch the tests.
	\end{itemize}
	
	It should be clear that both automata network ${\mathcal{A}}$ and specification $\mathcal{E}$ can be constructed in time ${O(k\cdot n^{O(1)})}$ from $G$. The construction is such that $G$ admits a dominating set of size $k$ if and only if $\mathcal{A} \models \mathcal{E}$: indeed, from the initialization imposed by ${\mathcal{E}}$ it would take at most $n$ steps for any of the two test to raise an error state, so ${\mathcal{E}}$ ensures that there exists an initial marking that encodes a valid dominating set of size (at most) $k$ as explained above.
\end{proof}

From Lemma~\ref{lem:kdominating} and $W[2]$-hardness of the $k$-\textsf{Dominating-Set} problem \cite{Downey_1995}, we immediately get the following corollary.

\begin{corollary}\label{cor:w2fixedalphabet}
  The \textbf{SPEC} problem with fixed degree and fixed alphabet and with treewidth as unique parameter is ${W[2]}$-hard.
\end{corollary}

A freezing automata network on a ${O(k)\times n^2}$-grid with alphabet $Q$ can be seen as a freezing automata network on a line of length ${n^2}$ with alphabet ${Q^{O(k)}}$.
One might therefore want to adapt the above result to show $W[2]$-hardness in the case where treewidth and degree are fixed while alphabet is the parameter. However, the specification which is part of the input, has an exponential dependence on the alphabet (a $t$-specification is of size ${O(n\cdot t^{|Q|})}$). Therefore FPT reductions are not possible when the alphabet is the parameter. We can circumvent this problem by considering a new variant of the \textbf{SPEC} problem where specification are given in a more succinct way through regular expressions. A regular ${(Q,V)}$-specification is a map from $V$ to regular expressions over alphabet $Q$. We therefore consider the problem \textbf{REGSPEC} which is the same as \textbf{SPEC} except that the specification must be a regular specification. With this modified settings, the construction of Lemma~\ref{lem:kdominating} can be adapted to deal with the alphabet as parameter.

\begin{corollary}\label{cor:w2fixedtreewidth}
  The \textbf{REGSPEC} problem with fixed degree and fixed treewidth and with alphabet as unique parameter is ${W[2]}$-hard.
\end{corollary}

\begin{proof}
  Using the construction of Lemma~\ref{lem:kdominating} and compressing the $k$ rows into a single one by enlarging the alphabet, we can construct a freezing automata network $\mathcal{A}^\alpha$ of alphabet ${Q^{k+2}}$ on a graph ${G''=(V'',E'')}$ which is a cycle of length $n^2$ (therefore of constant treewidth and constant degree) and that has the same behaviour with respect to $k$-dominating sets of the graph $G$ of the lemma. The local map at each node of $\mathcal{A}^\alpha$ can be described by a transition table of size ${O(Q^{3k})}$ so the global description is of size ${O(Q^{3k}\cdot n^2)}$. Noting that the specification produced in Lemma~\ref{lem:kdominating} is actually a regular specification of the form: ${v\in V'\mapsto Q_{i,v}Q_{e,v}^\ast}$ where $Q_{i,v}$ take care of the initialization and $Q_{e,v}$ is the subset of states with no error. We deduce that the corresponding regular specification for ${\mathcal{A}^\alpha}$ is of the form ${v\in V''\mapsto \bigl(Q_{i,v_1}\times\cdots\times Q_{i,v_{k+2}}\bigr)\bigl(Q_{e,v_1}\times\cdots\times Q_{e,v_{k+2}}\bigr)^\ast}$. Hence its size is ${O(|Q|^{O(k)}\cdot n^2)}$. The total size of the input produced for problem \textbf{REGSPEC} is therefore also ${O(|Q|^{O(k)}\cdot n^2)}$ and it can be produced in time ${O(|Q|^{O(k)}\cdot n^{O(1)})}$. This proves that the $k$-dominating set problem can be FPT reduced to \textbf{REGSPEC} with alphabet as parameter.
\end{proof}

\section{Hardness results for polynomial treewidth networks}

We say a family of graphs ${\mathcal{G}}$ has polynomial treewidth if the graphs of the family are of size at most polynomial in their treewidth, precisely: if there is a non-constant polynomial map $p_{\mathcal{G}}$ (with rational exponents in $(0,1)$) such that for any ${G=(V,E)\in\mathcal{G}}$ it holds ${\tw(G)\geq p_{\mathcal{G}}(|V|)}$. Moreover, we say the family is \emph{constructible} if there is a polynomial time algorithm that given $n$ produces a connex graph ${G_n\in\mathcal{G}}$ with $n$ nodes.
The following lemma is based on a polynomial time algorithm to find large perfect brambles in graphs \cite{KreutzerT10}. This structure allows to embed any digraph in an input graph with sufficiently large treewidth via path routing while controlling the maximum number of intersections per node of the set of paths.

\begin{lemma}[Subgraph routing lemma]
  \label{lem:routing}
  For any family $\mathcal{G}$ of graphs with polynomial treewidth, there is a polynomial map $p$ and a deterministic polynomial time algorithm that, given any graph ${G=(V,E)\in\mathcal{G}}$ and any digraph ${D=(V',E')}$ of maximum (in/out) degree $\Delta$ and size at most ${p(|V|)}$, outputs:
  \begin{itemize}
  \item a mapping ${\mu: V'\rightarrow V}$ such that, for each ${v\in V}$, ${\mu^{-1}(v)}$ contains at most two elements,
  \item a collection ${\mathcal{C}=(p_{e'})_{e'\in E'}}$ of paths connecting ${\mu(v'_1)}$ to ${\mu(v'_2)}$ for each ${(v'_1,v'_2)\in E'}$, and such that any node in $V$ belongs to at most $4\Delta$ paths from $\mathcal{C}$.
  \end{itemize}
\end{lemma}
\begin{proof}
	By \cite[Theorem 5.3]{KreutzerT10} there exists a polynomial map $p_1$ and a polynomial time algorithm that given a graph ${G=(V,E)\in\mathcal{G}}$ finds a \emph{perfect bramble} ${\mathcal{B}=(B_1,\ldots, B_k)}$ with ${k\geq p_1(p_{\mathcal{G}}(|V|))}$, \textit{i.e.} a list of connected subgraphs ${B_i\subseteq V}$ such that:
	\begin{enumerate}
		\item ${B_i\cap B_j\not=\emptyset}$ for all $i$ and $j$,
		\item for all ${v\in V}$ there are at most two elements of $\mathcal{B}$ that contain $v$.
	\end{enumerate}
	We set the polynomial map of the lemma to be $p=p_1\circ p_{\mathcal{G}}$ and consider any digraph ${D=(V',E')}$ of maximum (in/out) degree $\Delta$ and size at most ${p(|V|)}$. We suppose ${k=|V'|}$ (by forgetting some elements of $\mathcal{B}$) and reindex the element of $\mathcal{B}$ by $V'$. The map ${\mu: V'\rightarrow V}$ is constructed by picking some element ${\mu(v')\in B_{v'}}$ for all ${v'\in V'}$. The fact that any vertex $v\in V$ is contained in at most two elements of the bramble $\mathcal{B}$ ensures the first condition of the lemma on $\mu$. Now, for each ${(v'_1,v'_2)\in E'}$ we define a path from ${\mu(v'_1)}$ to ${\mu(v'_2)}$ as follows: let ${v\in B_{v'_1}\cap B_{v'_2}}$ (first property of perfect brambles) then choose a path from $\mu(v'_1)$ to $v$ inside $B_{v'_1}$ (which is connected) followed by a path from $v$ to $\mu(v'_2)$ inside $B_{v'_2}$. The collection of paths $\mathcal{C}$ thus defined is such that there are at most $2\Delta$ paths that start or end in ${\mu(v')}$ for any ${v'\in E}$. Moreover, for any ${v\in V}$, there are at most two elements of ${\mathcal{B}}$ that contain $v$, let's say ${B_{v'_1}}$ and ${B_{v_2'}}$. Then the only paths from $\mathcal{C}$ that can go through $v$ are those starting or ending at either $\mu(v_1')$ or $\mu(v_2')$, so they are at most $4\Delta$ in total.
\end{proof}

\begin{theorem}
  \label{thm:nilpotency}
  For any family ${\mathcal{G}}$ of constructible graphs of polynomial treewidth, the problem \emph{nilpotency} is coNP-complete.
\end{theorem}
\begin{proof}
	First, by Lemma~\ref{lem:trace-length}, a freezing automata networks with $n$ nodes is nilpotent if and only if ${F^{\lambda(n)}}$ is constant where ${\lambda(n)\in O(n)}$ is the concrete computable bound from the lemma. The nilpotency problem is therefore clearly coNP.
	
	We now describe a reduction from problem SAT. Given a formula with $n$ variables seen as a Boolean circuit of maximum input/output degree $2$ (of size polynomial in $n$), we first construct ${G=(V,E)\in\mathcal{G}}$ such that the DAG $G'=(V',E')$ associated to the circuit is of size at most $p(|V|)$ where $p$ is the polynomial map of Lemma~\ref{lem:routing}. Then, using Lemma~\ref{lem:routing}, we have a map ${\mu:V'\rightarrow V}$ and a collection $\mathcal{C}$ of paths in $G$ that represent an embedding of $G'$ inside $G$. The lemma gives a bound $8$ on the number of paths visiting a given node ${v\in V}$. Then each node will hold $8$ Boolean values, each one corresponding either a node ${v'\in V'}$ of the Boolean circuit or an intermediate node of a path from the collection $\mathcal{C}$. The alphabet is then ${Q=\{0,1\}^8\cup \{\bot\}}$ where $\bot$ is a special error state. In any configuration ${c\in Q^V}$, a node can be either in error state $\bot$, or it holds $8$ Boolean components. We then construct the local rule at each node ${v\in V}$ that give a precise fixed role to each such component: it either represent a node $v'\in V'$ such that ${\mu(v')=v}$, or an intermediate node in one of the paths from ${\mathcal{C}}$, or is unused (because not all vertices of $V$ have $8$ paths from $\mathcal{C}$ visiting them). The local rule at ${v\in V}$ is as follows:
	\begin{itemize}
		\item if in state $\bot$ or if some neighbors is in state $\bot$, it stays in or changes to $\bot$;
		\item it then make the following checks and let the state unchanged if they all succeed or changes to $\bot$ if at least one test fails:
		\begin{enumerate}
			\item check for any component corresponding to a node ${v'\in V'}$ that it holds the Boolean value ${g(x,y)}$ where $g$ is the Boolean gate associated to $v'$ in the the circuit and $x$ and $y$ are the Boolean values of the components corresponding to the vertex just before $v$ in the two paths $\rho_{e_1}$ and $\rho_{e_2}$ in $\mathcal{C}$ that arrive at $\mu(v')=v$. In the case where $g$ is a 'not' gate, there is only one input and in the case where $v'$ is an input of the circuit, there is no input and nothing is checked;
			\item moreover, if the gate corresponding to $v'$ is the output gate of the circuit, check that its Boolean value is $1$;
			\item check for any component corresponding to an intermediate node in some path from $\mathcal{C}$ that the Boolean value it holds is the same as that of the component corresponding to the predecessor in the path.
		\end{enumerate}
	\end{itemize}
	We claim that $F$ is not nilpotent if and only if the formula represented by the Boolean circuit is satisfiable. Indeed the configuration everywhere equal to $\bot$ is always a fixed point. It should be clear that if the formula is satisfiable then one can build a configuration corresponding to a valid computation of the circuit on a valid input which is a fixed point not containing state $\bot$. In this case we have two distinct fixed points and the automata network is not nilpotent. Conversely, suppose the the automata network is not nilpotent. Then it must possess a fixed point $c$ distinct from the all $\bot$ one. Indeed, all configurations of ${X=F^t(Q^V)}$ are fixed points for $t$ large enough (by the freezing condition) and if $F^t$ is not a constant map then $X$ must contain at least two elements. Moreover, the fixed point $c$ do not contain state $\bot$, because otherwise it would contain a state from ${Q\setminus\{\bot\}}$ at some node which has a neighbor in state $\bot$, which would contradict the fact that it is a fixed point according to the local rule. Then $c$ is a configuration where all checks made by the local rules are correct: said differently, $c$ contains the simulation of a valid computation of the Boolean circuit that outputs $1$. Therefore the Boolean formula is satisfiable and the reduction follows.
\end{proof}

When giving an automata network as input, the description of the local functions depends on the underlying graph (and in particular the degree of each node). However, some local functions are completely isotropic and blind to the number of neighbors and therefore can be described once for all graphs. This is the case of local functions that only depends on the set of states present in the neighborhood. Indeed, given a map ${\rho : Q\times 2^Q\rightarrow Q}$ and any graph $G=(V,E)$, we define the automata network on $G$ with local functions ${F_v:Q^{N(v)}\rightarrow Q}$ such that 
${F_v(c) = \rho\bigl(c(v),\{c(v_1),\ldots,c(v_k)\}\bigr)}$ where ${N(v)=\{v_1,\ldots,v_k\}}$ is the neighborhood of $v$ which includes $v$. We then say that the automata network is \emph{set defined} by $\rho$. We will prove the next two hardness results with a fixed set defined rule, showing that there is a uniform and universally hard rule on graphs of polynomial treewidth for predecessor and asynchronous reachability problems. 
The proof below uses again Lemma~\ref{lem:routing} to embed arbitrary circuits like in theorems above, but the difference here is that the circuit embedding is written in the configuration and is not hardwired into the local rule. Moreover, the reduction also uses $L(1,1)$ graph coloring \cite{Calamoneri2006} to deal with communication routing in a set defined rule in a similar way as in a radio network.

\begin{theorem}
  There exists a map ${\rho : Q\times 2^Q\rightarrow Q}$ such that for any family ${\mathcal{G}}$ of constructible graphs of polynomial treewidth and bounded degree, the problems \emph{predecessor} and \emph{asynchronous reachability} are both NP-complete when restricted to $\mathcal{G}$ and automata networks set-defined by $\rho$.
  \label{teo:predeasyncNP}
\end{theorem}
\begin{proof}
	These problems are clearly NP. For clarity of exposition we will construct a distinct map $\rho$ for each of the two problems. Then, by taking the disjoint union of the alphabets and merging the two rules with the additional condition that any node that sees both alphabets is left unchanged, we obtain a single map $\rho$ that is hard for both problems. Indeed, using the first alphabet only for predecessor problem inputs, we have the guarantee that the only possible pre-images must only use the first alphabet, hence the hardness follows for the combined rule. The same is true for asynchronous reachability using the second alphabet.
	
	Let's now describe ${\rho_1 : Q_1\times 2^{Q_1}\rightarrow Q_1}$ that set defines automata networks which have a NP-complete predecessor problem when restricted to $\mathcal{G}$. We describe it while showing the polynomial time reduction from SAT to the predecessor problem. Given a formula with $n$ variables seen as a Boolean circuit of maximum input/output degree $2$ (of size polynomial in $n$), we first construct ${G=(V,E)\in\mathcal{G}}$ such that the DAG $G'=(V',E')$ associated to the circuit is of size at most $p(|V|)$ where $p$ is the polynomial map of Lemma~\ref{lem:routing}. Then, using Lemma~\ref{lem:routing}, we have a map ${\mu:V'\rightarrow V}$ and a collection $\mathcal{C}$ of paths in $G$ that represent an embedding of $G'$ inside $G$. The lemma gives a bound $8$ on the number of paths visiting a given node ${v\in V}$. 
	Let's compute a vertex coloring ${\chi:V\rightarrow\{1,\ldots,k\}}$ of the square of $G$ with ${k\leq deg(G)^2+1}$ colors, \textit{i.e.} a vertex coloring of $G$ such that no pair of neighbors of a given node has the same color (this can be done in polynomial time by a greedy algorithm). To implement the routing of information along paths of $\mathcal{C}$ and the circuit simulation by $\rho_1$, the alphabet $Q_1$ holds $8k$ state components, and we will use configurations where each node $v\in V$ uses only components ${8\chi(v)}$ to ${8\chi(v)+7}$. These components can be seen as communication channels. Indeed, in such configurations, a node can distinguish the information going through up to $8$ distinct paths coming from each neighbor individually just by looking at the set of states present in the neighborhood (because no pair of neighbors can use the same channel). Apart from the routing of information through paths, the rule $\rho_1$ implements each gate $v'\in V'$ of the Boolean circuits inside node $\mu(v')$ of $G$. We think of paths from $\mathcal{C}$ as being part of the circuit with nodes that implement the identity map. For that purpose each state component in a node is associated to a descriptor that gives the type of gate to implement (input, identity, not, or, and, output) and the component numbers corresponding to input(s) of the gate (gates of type 'input' have no input).
	\newcommand\sok{\texttt{ok}}
	\newcommand\soff{\texttt{off}}
	Formally, a state component is given by ${S_1=\{0,1,\sok,\soff\}}$ where $0$ and $1$ are Boolean values, $\soff$ means that the component is unused and $\sok$ is a special transitory state used to check correctness of computations (see below). A descriptor component is given by ${D_1}$, a finite set used to code any possible combination of gate type and input component numbers (${|D_1|=6(8k)^2}$ is enough). Then the state set of $\rho_1$ is ${Q_1=(S_1\times D_1)^{8k}}$. In a given configuration, we say that a given node ${v\in V}$ reads value ${x\in\{0,1\}}$ on channel $i$ if there is a unique state in the neighborhood with a state component $i$ which is not $\soff$, and if this state component contains value $x$. In any other case, the value read on channel $i$ is undefined. The rule $\rho_1$ does the following:
	\begin{itemize}
		\item the $D_1$ component are never changed;
		\item state components in $\soff$ stay unchanged;
		\item any state component in $\sok$ becomes $\soff$;
		\item any state component in state ${x\in\{0,1\}}$ checks that $x$ is the correct output value of its gate type applied to the values read on the input channels given by its corresponding descriptor (in particular these input values must be defined). If it is the case, it becomes ${\sok}$, otherwise ${\soff}$. The only exception to this rule is the case of the gate of type ``output'' where we only change state to ${\sok}$ if $x=1$ and the computation check is correct, and change to $\soff$ in any other case.
	\end{itemize}
	We then build configuration $c\in Q_1^V$ for the predecessor problem as follows:
	\begin{itemize}
		\item input component numbers and gate types in $D$ components are set
		according to the Boolean circuit and the path collection
		${\mathcal{C}}$;
		\item all unused state components are marked as $\soff$;
		\item all used state components are marked $\sok$.
	\end{itemize}
	We claim that $c$ has a predecessor in one step (\textit{i.e.} ${F_{\rho_1}(y)=c}$ for some ${y\in Q_1^V}$) if and only if the SAT formula represented by the Boolean circuit is satisfiable. Indeed, the only possible predecessor configurations of $c$ are such that all used state component hold a Boolean value equal to the output value of the gate they code applied to their corresponding input Boolean values, and that the output gate holds value $1$. 
	
	We now describe ${\rho_2 : Q_2\times 2^{Q_2}\rightarrow Q_2}$ that set defines automata networks which have a NP-complete asynchronous reachability problem when restricted to $\mathcal{G}$. The construction is almost identical to $\rho_1$ and the reduction is again from SAT problems, but with the following modifications:
	\begin{itemize}
		\item the state component is now ${S_2=\{?,0,1,\sok,\soff\}}$ where the new state $?$ represents a pre-update standby state; in each state component, the possible state sequences are subsequences of either ${?\rightarrow \{0,1\}\rightarrow \sok \rightarrow \soff}$ or ${?\rightarrow \{0,1\}\rightarrow \soff}$ ;
		\item to each input gate of the Boolean circuit is attached a pre-input gate that serve as non-deterministic choice for input gates using asynchronous updates; the set $D_2$ is a modification of $D_1$ taking into account this new type of gates; the alphabet is then ${Q_2=(S_2\times D_2)^{8k}}$;
		\item the behavior of each state component depending on its type is as follows:
		\begin{itemize}
			\item pre-input components become $\sok$ if previously in state $?$ and $\soff$ in any other case;
			\item input components in state $?$ become either $0$ or $1$ depending on whether there corresponding pre-input component is in state $?$ or not;
			\item any other state component in state $?$ become ${x\in\{0,1\}}$ the output value of its gate type applied to the values read on the input channels given by its corresponding descriptor (in particular these input values must be defined). If in a state from ${\{0,1\}}$, it becomes ${\sok}$, and in any other case it becomes ${\soff}$. The only exception to this rule is the case of gates of type ``output'' where we only change state to ${\sok}$ if the current value is $1$, and change to $\soff$ if the current value is $0$.
		\end{itemize}
		When then define source configuration $c_0$ and destination configuration $c_1$ for the asynchronous reachability problem as follows. They both use the same circuit embedding like in $c$ above but with a pre-input attached to each input. In $c_0$ all unused components are in state $\soff$ and all used state components (including pre-inputs) are in state $?$. In $c_1$ all unused components are in state $\soff$ and all used state components are in state $\sok$. It should be clear that $c_1$ can be reached from $c_0$ if the formula associated to the Boolean circuit is satisfiable since either $0$ or $1$ can be produced at each input depending on whether the associated pre-input is update before the input update or not. Suppose now that $c_1$ can be reached from $c_0$ with some asynchronous update. First, all used state components except pre-inputs must follow either the sequence ${?\rightarrow 0 \rightarrow \sok}$ or ${?\rightarrow 1\rightarrow \sok}$. Therefore we can associate to each such component a unique Boolean value ($0$ or $1$ respectively) and the rule $\rho_2$ ensures that the Boolean value of each such component is the output value of its corresponding circuit gate applied on the Boolean value of its corresponding inputs. Moreover the output gate must have Boolean value $1$ so we deduce that the simulated circuit outputs $1$ on the particular choice of Boolean values of inputs. The reduction from SAT follows.
	\end{itemize}
	\end{proof}

In the remaining of this section, we focus on the prediction problem for families of graphs with polynomial treewidth. In particular, we are interested in deriving an analogous of Theorem \ref{teo:predeasyncNP} for prediction problem. Nevertheless, as a log-space or a NC reduction of some $\textbf{P}$-complete problem is required, most of the latter results that worked for Theorem \ref{teo:predeasyncNP} are not necessarily valid in this context as we only know that there exist polynomial time algorithms that compute certain needed structures. In order to face this task, our approach is  based in  slightly modifying the input of our prediction problem and then show that we can efficiently compute paths in a polynomial treewidth graph $G$.  The latter will allow us to show that we have an analogous of subgraph routing lemma (Lemma \ref{lem:routing}).   In particular, as it is not clear if the perfect bramble structures used in order to obtain Lemma \ref{lem:routing} are calculable in $\textbf{NC}$ or in log-space, we need to modify the problem in order to show that there exists a log-space reduction or an \textbf{NC} reduction of circuit value problem (\textsc{CVP})  in this particular variation of prediction, and thus that it is $\textbf{P}$-complete. More precisely, we add a perfect bramble of polynomial size to the input of \textsf{Prediction problem}. We call this modified version of prediction \textsf{Routed Prediction problem.} 

\begin{problem}[Routed prediction problem]\ 
	\begin{description}
		\item[Parameters: ] alphabet $Q$, family of graphs $\mathcal{G}$ of max degree $\Delta$
		\item[Input: ]\ 
		\begin{enumerate}
			\item a deterministic freezing automata network
			$\mathcal{A}=(G,F)$ on alphabet $Q$, with set of nodes
			$V$ with ${n=|V|}$ and $G\in\mathcal{G}$;
			\item an initial configuration $c\in Q^V$
			\item a node $v\in V$ and a ${(\{v\},Q,t)}$-specification $\mathcal{S}_v$ of length ${t\in\N}$
			\item A perfect bramble  $\mathcal{B} = (B_1,\hdots, B_p)$ in with $p = n^{\mathcal{O}(1)}$ in  $G$
		\end{enumerate}
		\item[Question: ] does the orbit of $c$ restricted to $v$ satisfies specification $\mathcal{S}_v$?
	\end{description}
\end{problem}

Now, having this latter problem in mind,  we slightly modify the definition of a constructible familly of graphs $\mathcal{G}$ of polynomial treewidth introduced at the begining of this section: we define a  routed collection of graphs of polynomial treewidth to the set $\mathcal{G} = \{(G_n,\mathcal{B}_n)\}_{n\in N}$ such that $G_n$ is an undirected connected graph of order $n$ and treewidth $\text{tw}(G_n) \geq p(n)$ and $\mathcal{B}_n$ is a perfect bramble such that $|\mathcal{B}_n| \geq p'(n)$ where $p$ and $p'$ are polynomials. We say a that  a routed collection of graphs of polynomial treewidth $\mathcal{G}$  is  log-constructible if there is a log-space algorithm that given $n$ produce the tuple $(G_n , \mathcal{B}_n) \in \mathcal{G}$.  
As we will be working with a log-constructible collection of routed graphs, we would like to say that we could have the result of Lemma \ref{lem:routing} in order to show the main result of this section. Nevertheless, in order to do that,  we need  to have a log-space or a NC algorithm computing the paths that we will be using for the proof of the main result.  More precisely, we need to compute the function $\mu$ and the collection of paths $\mathcal{C}$. Fortunately, in \cite[Theorem 5.3]{Reingold2008} it is shown that there exist a log-space algorithm that accomplish this task.
Finally, as in the proof of Theorem \ref{teo:predeasyncNP}, we need a proper coloring of the square graph $G^2$ in order to broadcast information through the paths in the collection $\mathcal{C}$ without encountering problems in the nodes that are in different paths at the same time. Fortunatly, we can do this in $\textbf{NC}$ as it is stated in \cite[Theorem 3]{goldberg1987parallel}.
We are now in condition of showing our main result concerning routed prediction problem:

\begin{theorem}
  \label{theo:pcompleteprediction}
  There exists a map ${\rho : Q\times 2^Q\rightarrow Q}$ such that  \textsf{Routed Prediction problem} is \textbf{P}-complete restricted to any family $\mathcal{G}$ of log-constructible routed collection of graphs of polynomial treewidth.
\end{theorem}
\begin{proof}
	We start by observing that prediction problem is in $\textbf{P}$.
	We also recall that in order to show the $\textbf{P}$-hardness of prediction problem, it suffices to show that there exist a \textbf{NC} reduction for the alternating monotone $2$ fan-in $2$ fan-out circuit value problem ($\textsc{AM2CVP}$), more precisely $ \textsc{AM2CVP}\leq^{\textbf{NC}^2}_{m} \textsc{PRED}_{\mathcal{G}}$  (see \cite{greenlaw1995limits} Theorem 4.2.2 and Lemma 6.1.2). Let $n,l \in \mathbb{N}$,  $C:\{0,1\}^n \to \{0,1\}^l$ a monotone alternating $2$ fan in $2$ fan out circuit, $x \in \{0,1\}^n$ and $o \in \{0,\hdots, l-1\}$ a fixed output of $C$. We call $C' = (V',E')$ to the underlying DAG defining $C$ and we fix $G \in \mathcal{G}$ where $\mathcal{G}$ is a log-constructible family of graphs with polynomial treewidth. We note that, by definition we can compute $G$ and a perfect bramble of size $p = n^{\mathcal{O}(1)}$ in log-space  and thus we can do the latter computations in  $\textbf{NC}^2$.  Now, we use $\mathcal{B}$ and Proposition \ref{prop:pathsLOG} in order to compute a mapping $\mu: V' \to V$ and a collection of paths $\mathcal{C}$ as in Lemma \ref{lem:routing}.   As we did in Theorem \ref{teo:predeasyncNP}, we use Proposition \ref{prop:colNC} in order to compute a $k$-proper coloring $\chi: V \to  \{1,\hdots, k\}$  of $G^2$ in $\textbf{NC}^2$ with $k = \Delta^2 +1 $ for $\Delta \in \mathbb{N}$ such that $\Delta(G) = \Delta$. From here we construct $\rho$ analogously  as we did for $\rho_1$  and $\rho_2$ in the proof of Theorem \ref{teo:predeasyncNP} but observing that now we have only $5$ type of gates as the circuit is monotone. We also consider  state component $S = \{0,1,\textbf{wait}, \textbf{off}\}$. Remember that the descriptor component assures that there won't be overlappings of the channels during broadcasting. We  map $x$  into a configuration $y  \in Q = \{S \times D\}^{8k}$ in the following way: 
	\begin{itemize}
		\item The $D$ component is assigned according to the structure of $C'$.
		\item For every input we assign a boolean value given by $x.$
		\item For every unused node we assign the state $\textbf{off}$.
		\item For every other node we assign the state $\textbf{wait}$.
	\end{itemize}
	The rule $\rho$ is defined in the following way:
	\begin{itemize}
		\item Every node in state $\textbf{off}$, $0$ or $1$ is fixed and does not change its state.
		
		\item Every node in state $\textbf{wait}$ reads the information of its neighbors and do the following depending on its type of gate:
		\begin{itemize}
			\item identity will take the value of its input
			\item AND will read its inputs: if both inputs are in $1$ it will change to $1$ and it will change to $0$ if it reads one neighbor in $0$. In any other case it will remain in $\textbf{wait}$
			\item OR  will read its inputs: if both inputs are in $0$ it will change to $0$ and it will change to $1$ if it reads one neighbor in $1$. In any other case it will remain in $\textbf{wait}$
		\end{itemize}
	\end{itemize}
	
	In order to show the desired result,  it will suffices to show the following simulation property: there exists $t \in \mathbb{N}$, $ t = n^{\mathcal{O}(1)}$ such that for every output $o \in V'$ we have $C(x)_o = (F^t(y)_{\mu(o)})|_S$ where $y\in Q^V$ is the configuration computed from $x$ as explained above.  In fact, if we have the latter property,  for some fixed output $o$, we define  $v = \mu(o)$ and $\mathcal{S}_v$ be a $(\{v\},Q,t)$-specification such that $\mathcal{S}_v = \{z \in \{0,1\}: z \not = y_v \text{ and } (F_\rho(y)^t|_v)|_S = z \} $ and then we can answer if the orbit of $y$  in time $t$ given by $F^t_{\rho}(y)$ satisfies $\mathcal{S}_v$ if and only if  we can answer if $C(x)_{o} = 1$ (and thus   $\textsc{AM2CVP}\leq^{\textbf{NC}^2}_{m} \textsc{PRED}$). We now show that the latter simulation property holds. In order to do that, we inductively check, that eventually, the orbit of $y$ will evaluate every layer of the circuit. We start by the input. Note that in one time step all the information is broadcasted through the different channels and through the paths given by $\mathcal{C}$. In a maximum of $L = n^{\mathcal{O}(1)}$  time steps (given by the longest path of $C'$) the last signal will arrive to a gate in the first layer. Note that, with the gates described above, signals arriving at different times do not change its output value as gates have a monotone behaviors on states with order ${\mathbf{wait}\leq 0\leq 1}$. Iteratively, we have maximum arriving times for signals of $L$ time steps for each layer and then, defining $t = L \times \text{deph}(C) = n^{\mathcal{O}(1)}$ and observing that output nodes will  will remain  constant once they have done a computation (when they change to a boolean value), we get the desire result. Therefore, $\textsc{AM2CVP}\leq^{\textbf{NC}^2}_{m} \textsc{PRED}$ holds and then, $\textsc{PRED}_{\mathcal{G}}$ is $\textbf{P}$-complete.
	
\end{proof}
\section{Discussion}

In this paper, we established the key role of treewidth and maximum degree in the computational complexity of freezing automata networks. We believe that our results can be extended in several ways.

First, our algorithm for the general model checking problem is not as efficient as known algorithms for specific sub-problems \cite{Buss_1987} and it would be interesting to establish hardness results in the NC hierarchy to make this gap more precise. In the same vein, our algorithm doesn't yield fixed parameter tractability results for any of the parameters (treewidth, degree, alphabet), and we wonder whether our hardness results in the framework of parameterized complexity \cite{Downey_2013} could be improved. We could also consider intermediate treewidth classes (non-constant but sub-polynomial). Concerning these complexity questions, we think that considering other (more restrictive) parameters like pathwidth could definitely help to obtain better bounds.

Besides, one might wonder whether the set of dynamical properties that are efficiently decidable on graphs of bounded degree and treewidth could be in fact much larger than what gives our model checking formalism. This question remains largely open, but we can already add ingredients in our formalism (for instance, a relational predicate representing the input graph structure). We however conjecture that there are NP-hard properties for freezing automata network on trees of bounded degree that can be expressed in the following language: first order quantification on configurations together with a reachability predicate (configuration $y$ can be reached from $x$ in the system).

Finally, we think that we can push our algorithm further and partly release the constraint on maximum degree (for instance allowing a bounded number of nodes of unbounded degree). This can however not work in the general model checking setting as shown in Remark~\ref{rem:nilpotencyontrees}. 

\bibliography{paper}

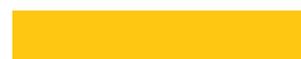
\end{document}